\documentclass{elsarticle}

\advance\hoffset by -1.5cm
\advance\voffset by -1cm

\usepackage{amssymb,amsthm,amsmath,enumerate,graphicx,hyperref,amsfonts,latexsym,url}
\usepackage[utf8]{inputenc}
\usepackage[T1]{fontenc}        
\usepackage[normalem]{ulem}
\usepackage[margin=2cm,font=small,labelfont=bf]{caption}

\usepackage{tikz}
\usetikzlibrary{snakes,shapes}
\tikzstyle{tre}=[circle,draw,minimum size=3mm,inner sep=1pt]%fill=green!50,
\tikzstyle{trep}=[circle,draw,minimum size=2mm,inner sep=0.1pt]%fill=green!50,
\newcommand{\etq}[1]{%
\draw (#1) node {\scriptsize $#1$};
}

\renewcommand{\leq}{\leqslant}
\renewcommand{\geq}{\geqslant}
\newcommand{\NN}{\mathbb{N}}

\newcommand{\TT}{\mathcal{T}}
\newcommand{\BT}{\mathcal{BT\!}}

\newcommand{\RR}{\mathbb{R}}

\theoremstyle{plain}
 \newtheorem{theorem}{Theorem}
 \newtheorem{lemma}[theorem]{Lemma}
 \newtheorem{corollary}[theorem]{Corollary}
 \newtheorem{proposition}[theorem]{Proposition}
 \theoremstyle{definition}
 \newtheorem{example}[theorem]{Example}
 
 \newtheorem{remark}[theorem]{Remark}

\makeatletter
\def\ps@pprintTitle{%
     \let\@oddhead\@empty
     \let\@evenhead\@empty
     \def\@oddfoot{}%
     \let\@evenfoot\@oddfoot}
\makeatother

\begin{document}
\begin{frontmatter}

\title{A balance index for phylogenetic trees based on rooted quartets}

\author[uib1]{Tom\'as M. Coronado}
\ead{t.martinez@uib.eu}
\author[uib1]{Arnau Mir}
\ead{arnau.mir@uib.eu}
\author[uib1]{Francesc Rossell\'o}
\ead{cesc.rossello@uib.eu}
\author[upc]{Gabriel Valiente}
\ead{valiente@cs.upc.edu}
\address[uib1]{Balearic Islands Health Research Institute (IdISBa) and Department of Mathematics and
  Computer Science,  University of the Balearic Islands,
  E-07122 Palma, Spain}
\address[upc]{Algorithms, Bioinformatics, Complexity and Formal Methods Research Group, Technical University of Catalonia, E-08034 Barcelona, Spain}

\begin{abstract}
We define a new balance index for rooted phylogenetic trees based on the symmetry of the evolutive history of every set of 4 leaves. This index makes sense for multifurcating trees and it can be computed in time linear in the number of leaves. We determine its maximum and minimum values for arbitrary and bifurcating trees, and we provide exact formulas for its expected value and variance on bifurcating trees under Ford's $\alpha$-model and Aldous' $\beta$-model and on arbitrary trees under the $\alpha$-$\gamma$-model. 
\end{abstract}
\end{frontmatter}

\section{Introduction}

One of the most broadly studied properties of the topology of rooted phylogenetic trees is their balance, that is, the tendency of the subtrees rooted at all children of any given node to have a similar shape. The main reason for this interest is that the balance of a rooted tree embodies the symmetry of the evolutive history it describes, and hence it reflects, at least to some extent, a feature of the forces that drove the evolution of the set of species represented in the tree; see Chapter 33 of \citep{fel:04}.

{The balance of a tree is usually quantified by means of \emph{balance  indices}. 
 The two most popular such indices are   \emph{Colless' index} \citep{Colless:82} for bifurcating trees, which is defined as the sum, over all internal nodes $v$, of the absolute value of the difference between the number of descendant leaves of the pair of children of $v$, and   \emph{Sackin's index} \citep{Sackin:72,Shao:90}, which is defined, for arbitrary trees, as the sum of the depths of all leaves in the tree. But many other balance indices have been proposed in the literature, like for instance, for bifurcating trees,  the variance of the depths
of the leaves \citep{Kirk,Sackin:72}, the sum of the reciprocals of
the orders of the rooted subtrees \citep{Shao:90}, and the number of
cherries \citep{cherries}, and, for arbitrary trees,  the \emph{total cophenetic index}~\citep{MRR}  and a generalization of the Colless index \citep{MRR:Plos}.
For more indices, see again Chapter 33 in the book by \citet{fel:04}. All these balance indices depend  only on the topology of the trees, not on the branch lengths or the actual labels on their leaves, although  the balance of time-stamped trees has also been considered by 
\citet{Dearlove}. This abundance of
balance indices is partly motivated by the advice given by \citet{Shao:90} to use more than one such index to quantify the
 balance of a tree, as well as by their use as tools to test stochastic models of evolution~\citep{BF,Egui,Kirk,Mooers97,Shao:90}; other properties of the shapes of phylogenetic trees used in this connection include the distribution of clades' sizes \citep{Zhu11,Zhu15} and the joint distribution of the numbers of  rooted subtrees of different types \citep{WuChoi:16}.

In this paper we propose a new balance index, the \emph{rooted quartet index}. To define it, we associate to each 4-tuple of different leaves of the tree $T$ a value that quantifies the symmetry of the joint evolution of the species they represent,  in the sense that it grows with the number of isomorphisms of the restriction of $T$ to them (the \emph{rooted quartet} they define), and then we add up these values over all 4-tuples of different leaves of $T$.
 The idea behind the definition of this balance index is that a highly symmetrical evolutive process should give rise to symmetrical evolutive histories of many small subsets of taxa. In terms of phylogenetic trees, this leads us to expect that, the most symmetrical a phylogenetic tree is, the most symmetrical will be its restrictions to subsets of leaves of a fixed cardinality. Since the smallest number of leaves yielding enough different tree topologies to allow a meaningful  comparison of their symmetry is 4,  we  assess the balance of a tree by  measuring the symmetry of all its rooted quartets and adding up the results.
 And indeed, in Section \ref{sec:maximin} below we shall find the trees with maximum and minimum values of our rooted quartet index in both the arbitrary and the bifurcating cases, and it will turn out that the minimum value is reached exactly at the combs (see Fig. \ref{fig:exs}.(a)), which are usually considered the least balanced trees, and the maximum value is reached, in the arbitrary case, exactly at the rooted stars (see Fig. \ref{fig:exs}.(b))  and, in the bifurcating case, exactly at the maximally balanced trees (cf. Fig. \ref{fig:maxbal} ), which in both cases are considered the most balanced trees.
 
Besides taking its maximum and minimum values at the expected trees, other important features of our index are that it can be easily computed in linear time and that its mean value and variance can be explicitly computed on any probabilistic model of phylogenetic trees satisfying two natural conditions: independence under relabelings and sampling consistency. This allows us to provide these values for two well-known probabilistic models of bifurcating phylogenetic trees, Ford's $\alpha$-model~\citep{Ford1} and Aldous' $\beta$-model~\citep{Ald1}, which include as specific instances the Yule~\citep{Harding71,Yule} and the uniform~\citep{CS,Rosen78,cherries} models, as well as for Chen-Ford-Winkel's $\alpha$-$\gamma$-model of multifurcating trees~\citep{Ford2}. To our knowledge, this is the first shape index for which closed formulas for the expected value and the variance under the $\alpha$-$\gamma$-model have been provided.}

The rest of this paper is organized as follows. In the next section we introduce the basic notations and facts on phylogenetic trees that will be used in the rest of the paper, and we recall several preliminary results on probabilistic models of phylogenetic trees, proving those results for which we have not been able to find a suitable reference in the literature.
Then, in Section 3, we define our rooted quartet index $\mathit{rQI}$ and we establish its basic properties. In Section 4 we compute its maximum and minimum values, and finally, in Section 5, we compute its expected value and variance under different probabilistic models. This paper is accompanied by the GitHub page~\url{https://github.com/biocom-uib/Quartet_Index} containing a set of Python scripts that perform several computations related to this index.

\section{Preliminaries}

\subsection{Notations and conventions}

In this paper, by an (\emph{unlabeled})  \emph{tree} we mean a  rooted tree without out-degree 1 nodes. As it is usual, we understand such a tree as a directed graph, with its arcs pointing away from the root. A tree is \emph{bifurcating}  when all its internal nodes have out-degree 2; when we want to emphasize that a tree need not be bifurcating, we shall call it \emph{multifurcating}. We shall denote by $L(T)$ the set of leaves of a tree $T$,  by $V_{int}(T)$ its set of internal nodes, and by $\mathrm{child}(u)$ the set of \emph{children} of an internal node $u$, that is, those nodes $v$ such that $(u,v)$ is an arc in $T$.  We shall always consider two isomorphic trees as equal, and we shall denote by $\TT^*_n$ and $\BT^*_n$ the sets of (isomorphism classes of) multifurcating trees and of bifurcating trees with $n$ leaves, respectively. 

A  \emph{phylogenetic tree} on a set $\Sigma$ is a tree with its leaves bijectively labeled in $\Sigma$. An \emph{isomorphism} of phylogenetic trees is an isomorphism of trees that preserves the leaves' labels. To simplify the language, we shall always identify a leaf of a phylogenetic tree  with its label and we shall say that two isomorphic phylogenetic trees ``are the same''. 
We shall denote by $\TT(\Sigma)$ and $\BT(\Sigma)$ the sets of (isomorphism classes of)  multifurcating phylogenetic trees and of bifurcating phylogenetic trees on $\Sigma$, respectively.    If $\Sigma$ and $\Sigma'$ are any two sets of labels of the same cardinality, say $n$, then any bijection $\Sigma\leftrightarrow\Sigma'$ extends in a natural way to bijections $\TT(\Sigma)\leftrightarrow \TT(\Sigma')$ and $\BT(\Sigma)\leftrightarrow \BT(\Sigma')$. When the specific set of labels $\Sigma$ is irrelevant and only its cardinality matters, we shall write  $\TT_n$ and $\BT_n$ (with $n=|\Sigma|$) instead of $\TT(\Sigma)$ and $\BT(\Sigma)$, and we shall identify   $\Sigma$ with the set  $[n]=\{1,2,\ldots,n\}$.
If $|\Sigma|=n$, there exists a forgetful mapping $\pi:\TT(\Sigma)\to \TT_n^*$ that sends every phylogenetic tree $T$ on $\Sigma$ to its underlying unlabeled tree: we shall call  $\pi(T)$ the \emph{shape} of $T$. We shall write $T_1\equiv T_2$ to denote that two phylogenetic trees $T_1,T_2$ (possibly on different sets of labels of the same cardinality) have the same shape.

We shall represent trees and phylogenetic trees by means of their usual Newick format,\footnote{See \url{http://evolution.genetics.washington.edu/phylip/newicktree.html}} although we shall omit the ending mark ``;'' in order not to confuse it in the text with a semicolon punctuation mark. In the case of (unlabeled) trees, we shall denote the leaves with $*$ symbols.

Given two nodes $u,v$ in a tree $T$, we say that $v$ is a  \emph{descendant} of $u$, and also that $u$ is an  \emph{ancestor} of $v$,
when there exists a path from $u$ to $v$ in   $T$; this, of course, includes the case of the stationary path from a node $u$ to itself, and hence, in this context, we shall use the adjective \emph{proper} to mean that $u\neq v$. 
Given a node $v$ of a  tree $T$, the \emph{subtree $T_v$ of $T$ rooted at $v$} is the subgraph of $T$ induced by the  descendants of $v$. We shall denote by $\kappa_T(v)$, or simply by $\kappa(v)$ if $T$ is implicitly understood, the number of leaves of $T_v$.

Given a tree $T$ and a subset $X\subseteq L(T)$, the \emph{restriction} $T(X)$ of $T$ to $X$ is the tree  obtained by first taking the subgraph of $T$ induced by all the ancestors of leaves in $X$ and then {suppressing} its out-degree 1 nodes. By \emph{suppressing} a node $u$ with out-degree 1  we mean that  if  $u$ is the root,  we  remove it together with the arc incident to it, and, if $u$ is not the root and if $u'$ and $u''$ are, respectively, its parent and its child, then we remove the node $u$ and the arcs $(u',u),(u,u'')$ and we replace them by a new arc $(u',u'')$.
For every $Y\subseteq L(T)$, the tree $T(-Y)$ obtained by \emph{removing} $Y$ from $T$ is nothing but the restriction $T(L(T)\setminus Y)$. If $T$ is a phylogenetic tree on a set $\Sigma$ and $X\subseteq \Sigma$,  the restrictions $T(X)$ and $T(-X)$ are phylogenetic trees on $X$ and $\Sigma\setminus X$, respectively.

A \emph{comb} is a bifurcating phylogenetic tree such that all its internal nodes have a leaf child: see Fig. \ref{fig:exs}.(a).
All combs with the same number $n$ of leaves have the same shape, and 
we shall generically denote them (as well as their shape in $\TT_n^*$)  by $K_n$. A \emph{rooted star} is a phylogenetic tree all of whose leaves are children of the root: see Fig. \ref{fig:exs}.(b).  For every set $\Sigma$, there is only one rooted star on $\Sigma$, and  if $|\Sigma|=n$, we shall generically denote it (as well as its shape)  by $S_n$.

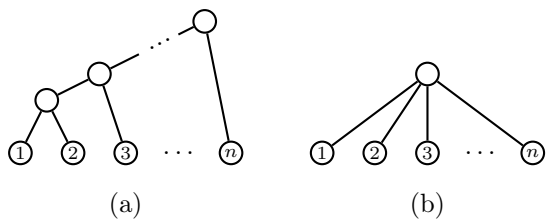
\begin{figure}[htb]
\begin{center}
\begin{tikzpicture}[thick,>=stealth,scale=0.35]
\draw(0,0) node [tre] (1) {};  \etq 1
\draw(2,0) node [tre] (2) {};  \etq 2
\draw(4,0) node [tre] (3) {};  \etq 3
\draw(6,0) node  {$\ldots$};  
\draw(8,0) node [tre] (n) {};  \etq n
\draw(1,2) node[tre] (a) {};
\draw(3,3) node[tre] (b) {};
\draw(7,5) node[tre] (r) {};
\draw(5,4) node  {.};
\draw(5.3,4.15) node  {.};
\draw(5.6,4.3) node  {.};
\draw (a)--(1);
\draw (a)--(2);
\draw (b)--(3);
\draw (b)--(a);
\draw (b)-- (4.5,3.75);
\draw (r)-- (6,4.5);
\draw (r)-- (n);
\draw(4,-2) node {(a)};
\end{tikzpicture}
\qquad
\begin{tikzpicture}[thick,>=stealth,scale=0.35]
\draw(0,0) node [tre] (1) {};  \etq 1
\draw(2,0) node [tre] (2) {};  \etq 2
\draw(4,0) node [tre] (3) {};  \etq 3
\draw(6,0) node  {$\ldots$};  
\draw(8,0) node [tre] (n) {};  \etq n
\draw(4,3) node[tre] (r) {};
\draw  (r)--(1);
\draw  (r)--(2);
\draw  (r)--(3);
\draw  (r)--(n);
\draw(4,-2) node {(b)};
\end{tikzpicture}
\end{center}
\caption{\label{fig:exs} 
(a) A  comb $K_n$. (b) A rooted star $S_n$.}
\end{figure}

Given $k\geq 2$ phylogenetic trees $T_1,\ldots,T_k$, with every $T_i\in \TT(\Sigma_i)$ and the sets of labels $\Sigma_i$ pairwise disjoint,
 their \emph{root join} is   the phylogenetic tree $T_1\star T_2\star\cdots \star T_k$ on $\bigcup_{i=1}^k\Sigma_i$ obtained by connecting the roots of (disjoint copies of) $T_1,\ldots,T_k$ to a new common root $r$; see Fig. \ref{fig:starcons}.  If $T_1,\ldots,T_k$ are unlabeled trees, a similar construction yields a tree $T_1\star\cdots \star T_k$.
 
 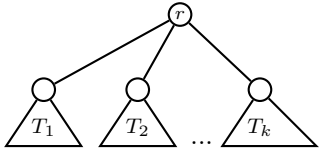
\begin{figure}[htb]
\begin{center}
\begin{tikzpicture}[thick,>=stealth,scale=0.25]
\draw(0,0) node[tre] (z1) {}; 
\draw (z1)--(-2,-3)--(2,-3)--(z1);
\draw(0,-2) node  {\footnotesize $T_1$};
\draw(5,0) node[tre] (z2) {}; 
\draw (z2)--(3,-3)--(7,-3)--(z2);
\draw(5,-2) node  {\footnotesize $T_2$};
\draw(8,-2.8) node {.}; 
\draw(8.4,-2.8) node {.}; 
\draw(8.8,-2.8) node {.}; 
\draw(11.5,0) node[tre] (z3) {}; 
\draw (z3)--(9.5,-3)--(14.5,-3)--(z3);
\draw(11.5,-2) node  {\footnotesize $T_k$};
\draw(7.25,4) node[tre] (r) {}; \etq r
\draw (r)--(z1);
\draw (r)--(z2);
\draw (r)--(z3);
\end{tikzpicture}
\end{center}
\caption{\label{fig:starcons} The root join  $T_1\star \cdots \star T_k$.}
\end{figure}

Let $T$ be a bifurcating  tree. For every $v\in V_{int}(T)$, say with children $v_1,v_2$,  the \emph{balance value} of $v$ is $bal_T(v)=|\kappa(v_1)-\kappa(v_2)|$. An internal node $v$ of  $T$ is \emph{balanced} when $bal_T(v)\leq 1$. So, a node $v$  with children $v_1$ and $v_2$ is  balanced if, and only if, $\{\kappa(v_1),\kappa(v_2)\}=\{\lfloor \kappa(v)/2\rfloor,\lceil \kappa(v)/2\rceil\}$.
We shall say that a bifurcating  tree $T$ is \emph{maximally balanced} when all its internal nodes are balanced. Recursively, a bifurcating  tree is maximally balanced when its root is balanced and the subtrees rooted at the children of the root are both maximally balanced.
This implies that,  for any fixed number $n$ of nodes, there is only one maximally balanced tree in $\BT_n^*$;  see Section 2.1 in \citep{MRR}. {Fig.~\ref{fig:maxbal} depicts the maximally balanced trees with $n=6,7,8$ leaves. When $n$ is a power of 2, the maximally balanced bifurcating  tree with $n$ leaves is the \emph{fully symmetric bifurcating tree}, where, for each internal node, the pair of subtrees rooted at its children  are  isomorphic; see again Fig.~\ref{fig:maxbal} for $n=8$.

\begin{figure}[htb]
\begin{center}
\begin{tikzpicture}[thick,>=stealth,scale=0.22]
\draw(0,0) node [trep] (1) {};  %\etq 1
\draw(2,0) node [trep] (2) {};  %\etq 2
\draw(4,0) node [trep] (3) {};  %\etq 3
\draw(6,0) node [trep] (4) {};  %\etq 4
\draw(8,0) node [trep] (5) {};  %\etq 5
\draw(10,0) node [trep] (6) {};  %\etq 6
\draw(1,2) node[trep] (a) {};
\draw(9,2) node[trep] (c) {};
\draw(2,4) node[trep] (b) {};
\draw(8,4) node[trep] (d) {};
\draw(5,6) node[trep] (r) {};
\draw  (a)--(1);
\draw  (a)--(2);
\draw  (b)--(a);
\draw  (b)--(3);
\draw  (c)--(5);
\draw  (c)--(6);
\draw  (d)--(4);
\draw  (d)--(c);
\draw  (r)--(b);
\draw  (r)--(d);
\end{tikzpicture}
\quad
\begin{tikzpicture}[thick,>=stealth,scale=0.22]
\draw(0,0) node [trep] (1) {};  %\etq 1
\draw(2,0) node [trep] (2) {}; % \etq 2
\draw(4,0) node [trep] (3) {};  %\etq 3
\draw(6,0) node [trep] (4) {};  %\etq 4
\draw(8,0) node [trep] (5) {};  %\etq 5
\draw(10,0) node [trep] (6) {}; % \etq 6
\draw(12,0) node [trep] (7) {}; % \etq 7
\draw(1,2) node[trep] (a) {};
\draw(5,2) node[trep] (b) {};
\draw(3,4) node[trep] (c) {};
\draw(11,2) node[trep] (d) {};
\draw(9,4) node[trep] (e) {};
\draw(6,6) node[trep] (r) {};
\draw  (a)--(1);
\draw  (a)--(2);
\draw  (b)--(3);
\draw  (b)--(4);
\draw  (c)--(a);
\draw  (c)--(b);
\draw  (d)--(6);
\draw  (d)--(7);
\draw  (e)--(d);
\draw  (e)--(5);
\draw  (r)--(c);
\draw  (r)--(e);
\end{tikzpicture}
\quad
\begin{tikzpicture}[thick,>=stealth,scale=0.22]
\draw(0,0) node [trep] (1) {};  %\etq 1
\draw(2,0) node [trep] (2) {};  %\etq 2
\draw(4,0) node [trep] (3) {};  %\etq 3
\draw(6,0) node [trep] (4) {};  %\etq 4
\draw(8,0) node [trep] (5) {};  %\etq 5
\draw(10,0) node [trep] (6) {}; % \etq 6
\draw(12,0) node [trep] (7) {}; % \etq 7
\draw(14,0) node [trep] (8) {}; % \etq 8
\draw(1,2) node[trep] (a) {};
\draw(5,2) node[trep] (b) {};
\draw(3,4) node[trep] (c) {};
\draw(9,2) node[trep] (d) {};
\draw(13,2) node[trep] (e) {};
\draw(11,4) node[trep] (f) {};
\draw(7,6) node[trep] (r) {};
\draw  (a)--(1);
\draw  (a)--(2);
\draw  (b)--(3);
\draw  (b)--(4);
\draw  (c)--(a);
\draw  (c)--(b);
\draw  (d)--(5);
\draw  (d)--(6);
\draw  (e)--(7);
\draw  (e)--(8);
\draw  (f)--(d);
\draw  (f)--(e);
\draw  (r)--(c);
\draw  (r)--(f);
\end{tikzpicture}
\end{center}
\caption{\label{fig:maxbal} 
The shapes of the maximally balanced trees with 6, 7, and 8 leaves. The tree with 8 leaves is fully symmetric.}
\end{figure}
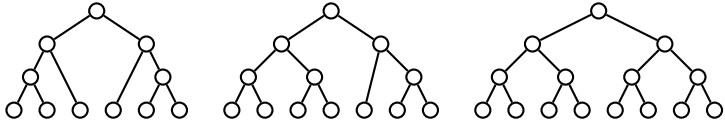
}

\subsection{Probabilistic models}
A \emph{probabilistic model of phylogenetic trees} $P_{n}$, $n \geq 1$,  is a family of probability mappings $P_n:\TT_n\to [0,1]$, each one sending each phylogenetic tree in $\TT_n$ to its probability under this model.  Every such a probabilistic model of phylogenetic trees  $P_{n}$ induces
a \emph{probabilistic model of trees}, that is, a family of probability mappings $P^*_n:\TT^*_n\to [0,1]$, by defining the probability of a tree   as the sum of the probabilities of all phylogenetic trees in $\TT_n$ with that shape:
$$
P^*_n(T^*)=\sum_{T\in \TT_n\atop \pi(T)=T^*} P_n(T).
$$
If $|\Sigma|=n$, then $P_n:\TT_n\to [0,1]$ induces also a probability mapping $P_{\Sigma}$ on $\TT(\Sigma)$ through the bijection $\TT_{\Sigma}\leftrightarrow \TT_n$ induced by a given bijection $\Sigma\leftrightarrow [n]$. 

A \emph{probabilistic model of bifurcating phylogenetic trees} is a probabilistic model  of  phylogenetic trees $P_n$ such that $P_n(T)= 0$ for every $T\in \TT_n\setminus \BT_n$.

A probabilistic model  of phylogenetic trees $P_{n}$ is  \emph{shape invariant} (or \emph{exchangeable}, according to \citet{Ald1}) when, for every $T,T'\in \TT_n$, if $T\equiv T'$, then $P_n(T)=P_n(T')$. In this case, for every $T^*\in \TT_n^*$ and for every $T\in \pi^{-1}(T^*)$,
$$
P_n^*(T^*)=\big|\{T'\in \TT_n: \pi(T')=T^*\}\big|\cdot P_n(T).
$$
Conversely, every probabilistic model of trees $P^*_{n}$ defines a  shape invariant probabilistic model  of phylogenetic trees $P_{n}$ by means of
\begin{equation}
P_n(T)=\frac{P_n^*(\pi(T))}{\big|\{T'\in \TT_n: T'\equiv T\}\big|}.
\label{eq:pstar}
\end{equation}
Notice that if $P_{n}$ is  {shape invariant}, then, for every set of labels $\Sigma$, say, with $|\Sigma|=n$, the probability mapping $P_{\Sigma}:\TT(\Sigma)\to [0,1]$ induced by the mapping $P_{n}$ does not depend on the specific bijection $\Sigma\leftrightarrow [n]$ used to define it.

A probabilistic model  of phylogenetic trees $P_{n}$ is \emph{sampling consistent} \citep{Ald1} (or also \emph{deletion stable}, according to \citet{Ford1}) when, for
every $n\geq 2$, if we choose a tree $T\in \TT_n$ with probability distribution $P_n$ and we remove its leaf $n$, the resulting tree is obtained with  probability distribution $P_{n-1}$; formally, when, for every $n\geq 2$ and for every $T_0\in \TT_{n-1}$,
$$
P_{n-1}(T_0)=\hspace*{-2.5ex}\sum_{T\in \TT_n\atop T(-n)=T_0}\hspace*{-2.5ex} P_n(T).
$$
It is straightforward to prove, by induction on $n-m$ and using that, for every $T\in \TT_n$ and for every $1\leq m< n$, the restriction of $T(-n)$ to $[m]$ is simply $T([m])$, that this condition is equivalent to the following: $P_{n}$ is {sampling consistent} when, for every $n\geq 2$, for every $1\leq m< n$, and for every $T_0\in \TT_m$,
\begin{equation}
P_m(T_0)=\hspace*{-2.5ex}\sum_{T\in \TT_n\atop T([m])=T_0}\hspace*{-2ex} P_n(T).
\label{eq:scphylo}
\end{equation}
It is also easy to prove that if $P_{n}$ is sampling consistent \emph{and} shape invariant, so that the probability of a phylogenetic tree is not affected by permutations of its leaves, then, for every $n\geq 2$, for every $\emptyset\neq X\subsetneq [n]$, say, with $|X|=m$, and for every $T_0\in \TT(X)$,
$$
P_X(T_0)=\hspace*{-2.5ex}\sum_{T\in \TT_n\atop T(X)=T_0}\hspace*{-2.5ex} P_n(T).
$$
(where $P_X$ stands for the probability mapping on $\TT(X)$ induced by $P_{m}$ through any bijection $X\leftrightarrow \big[m\big]$).

A probabilistic model  of  trees $P^*_{n}$ is \emph{sampling consistent} when, for every $n\geq 2$, 
 if we choose a tree $T\in \TT_n^*$ with probability distribution $P^*_n$ and  a leaf $x\in L(T)$ equiprobably and if we remove $x$ from $T$, the resulting tree is obtained with  probability distribution $P_{n-1}^*$:  formally, when, for
every $n\geq 2$ and for every $T_0\in \TT^*_{n-1}$,
$$
P^*_{n-1}(T_0) =\sum_{T\in \TT_n^*}\frac{\big|\{x\in L(T): T(-x)=T_0\}\big|}{n}\cdot P^*_n(T).
$$

We prove now several lemmas on probabilistic models that will be used in Section \ref{sec:QI}. The first lemma provides an extension of equation (\ref{eq:scphylo}) to trees; we include it because we have not been able to find a suitable reference for it in the literature. In it, and henceforth, $\mathcal{P}_k(X)$ denotes the set of all subsets of cardinality $k$ of $X$.

\begin{lemma}\label{lem:sampcons1}
A probabilistic model  of  trees $P^*_{n}$ is sampling consistent if, and only if, for every $n\geq 2$, for every $1\leq m< n$, and for every $T_0\in \TT^*_m$,
$$
P^*_m(T_0) = \sum_{T\in \TT_n^*}\frac{\big|\{X\in \mathcal{P}_m(L(T)): T(X)=T_0\}\big|}{\binom{n}{m}}\cdot P^*_n(T).
$$
\end{lemma}

\begin{proof} 
The ``if'' implication is obvious. As far as the ``only if'' implication goes, we prove by  induction on $n-m$ that if $P^*_{n}$ is sampling consistent, then, for every $T_0\in \TT^*_m$,
$$
P^*_m(T_0) = \sum_{T_n\in \TT_n^*}\frac{\big|\{X\in \mathcal{P}_{n-m}(L(T_n)): T_n(-X)=T_0\}\big|}{\binom{n}{m}}\cdot P^*_n(T_n).
$$
The starting case $m=n-1$ is the sampling consistency property. Assume now that this equality holds for $m$ and let $T_0\in \TT^*_{m-1}$. 
Then
$$
\begin{array}{l}
P_{m-1}^*(T_0) \displaystyle =\sum_{T_{m}\in \TT_{m}^*}\frac{\big|\{x\in L(T_{m}): T_m(-x)=T_0\}\big|}{m}\cdot P^*_{m}(T_{m})\\
\qquad\mbox{(by the sampling consistency)}\\[1ex]
 \displaystyle =\sum_{T_{m}\in \TT_{m}^*}\Bigg(\frac{\big|\{x\in L(T_{m}): T_m(-x)=T_0\}\big|}{m}\\[1ex]
\hspace*{2cm}\displaystyle \cdot \sum_{T_n\in \TT_n^*} \frac{\big|\{X\in \mathcal{P}_{n-m}(L(T_{n})): T_{n}(-X)=T_{m}\}\big|}{\binom{n}{m}}\cdot P_n^*(T_n)\Bigg)\\[1ex]
\qquad\mbox{(by the induction hypothesis)}\\[1ex]
 \displaystyle =\sum_{T_{m}\in \TT_{m}^*}\sum_{T_n\in \TT_n^*} \Bigg(\frac{\big|\{x\in L(T_{m}): T_m(-x)=T_0\}\big|}{m}\\[1ex]
\hspace*{2.5cm}\displaystyle \cdot  \frac{\big|\{X\in \mathcal{P}_{n-m}(L(T_{n})): T_{n}(-X)=T_{m}\}\big|}{\binom{n}{m}}\Bigg)\cdot P_n^*(T_n)\\[1ex]
 \displaystyle =\hspace*{-1.5ex}\sum_{T_n\in \TT_n^*}\hspace*{-1.5ex}
\frac{\big|\{(X,x)\in \mathcal{P}_{n-m}(L(T_{n}))\times (L(T_n)\setminus X): (T_{n}(-X))(-x)=T_0\}\big|}{m\cdot\binom{n}{m}}\cdot P_n^*(T_n)\\[2ex]
 \displaystyle =\hspace*{-1.5ex}\sum_{T_n\in \TT_n^*}\hspace*{-1.5ex}
\frac{\big|\{(X,x)\in \mathcal{P}_{n-m}(L(T_{n}))\times (L(T_n)\setminus X): (T_{n}(-(X\cup\{x\}))=T_0\}\big|}{m\cdot\binom{n}{m}}\cdot P_n^*(T_n)\\[2ex]
 \displaystyle =\sum_{T_n\in \TT_n^*}
\frac{(n-m+1) \big|\{Y\in \mathcal{P}_{n-m+1}(L(T_{n})): T_{n}(-Y) =T_0\}\big|}{m\cdot\binom{n}{m}}\cdot P_n^*(T_n)\\[2ex]
 \displaystyle =\sum_{T_n\in \TT_n^*}
\frac{\big|\{Y\in \mathcal{P}_{n-m+1}(L(T_{n})): T_{n}(-Y) =T_0\}\big|}{\binom{n}{m-1}}\cdot P_n^*(T_n)
\end{array}
$$
which proves the inductive step. \qed
\end{proof}

\begin{lemma}\label{lem:prevT}
Let $P_{n}$ be a shape invariant probabilistic model  of phylogenetic trees. For every $T_{n-1},T_{n-1}'\in \TT_{n-1}$, if $T_{n-1}\equiv T_{n-1}'$, then 
$$
\sum_{T_n\in \TT_n\atop T_n(-n)=T_{n-1}}\hspace*{-3ex} P_n(T_n)=\hspace*{-3ex} \sum_{T'_n\in \TT_n\atop T'_n(-n)=T'_{n-1}}\hspace*{-3ex} P_n(T'_n).
$$
\end{lemma}

\begin{proof}
Let $T_{n-1}^*=\pi(T_{n-1})=\pi(T'_{n-1})$ and let $f:T_{n-1}\to T'_{n-1}$ be an isomorphism of unlabeled trees, which exists because $T_{n-1}$ and $T_{n-1}'$ are both isomorphic as unlabeled trees to their shape $T_{n-1}^*$.
For every $T\in \TT_{n-1}$, let 
$$
E_n(T)=\{T_n\in \TT_n: T_n(-n)=T\}.
$$
 Each $T_n$ in $E_n(T_{n-1})$ is obtained by adding a leaf $n$ to $T_{n-1}$ as a new child either to an internal node, or to a new node obtained by splitting an arc into two consecutive arcs, or to a new bifurcating root (whose other child would be the old root). 
This entails the existence of a shape preserving bijection 
$$
\Phi: E_n(T_{n-1})\rightarrow E_n(T'_{n-1})
$$
that sends each $T_n\in E_n(T_{n-1})$ to the phylogenetic tree $\Phi(T_{n})$ obtained  by adding the leaf $n$ to $T_{n-1}'$ at the place corresponding through the isomorphism $f$ to the place where it has been added to $T_{n-1}$. Then, since $P_{n}$ is shape invariant,
$$
\sum_{T_n\in E(T_{n-1})}\hspace*{-3ex} P_n(T_n)
=\hspace*{-3ex} \sum_{T_n\in E(T_{n-1})}\hspace*{-3ex} P_n(\Phi(T_{n}))=\hspace*{-3ex}
\sum_{T'_n\in E(T'_{n-1})}\hspace*{-3ex} P_n(T'_n)
$$
%
%
%
%
%
%
%Therefore, for every $T_n^*\in \TT_n^*$, both $E_n(T_{n-1})$ and $E_n(T'_{n-1})$ contain the same number of phylogenetic trees with shape $T_n^*$. Since $(P_{n})_n$ is shape invariant, this implies 
%$$
%\begin{array}{l}
%\displaystyle \sum_{T_n\in E_n(T_{n-1})} P_n(T_n)=\sum_{T_n^*\in \TT_n}\sum_{T_n\in E_n(T_{n-1})\atop \pi(T_n)=T_n^*} P_n(T_n)\\
%\quad \displaystyle=
%\sum_{T_n^*\in \TT_n} \big|\{T_n\in E_n(T_{n-1}): \pi(T_n)=T^*_n\}\big|\cdot \frac{P_n^*(T_n^*)}{\big|\{T_n\in \TT_n: \pi(T_n)=T^*_n\}\big|}\\
%\quad \displaystyle=\sum_{T_n^*\in \TT_n} \big|\{T_n\in E_n(T'_{n-1}): \pi(T_n)=T^*_n\}\big|\cdot \frac{P_n^*(T_n^*)}{\big|\{T_n\in \TT_n: \pi(T_n)=T^*_n\}\big|}\\
%\quad \displaystyle=\sum_{T_n^*\in \TT_n}\sum_{T_n\in E_n(T'_{n-1})\atop \pi(T_n)=T_n^*} P_n(T_n)=\sum_{T_n\in E_n(T'_{n-1})} P_n(T'_n)
%\end{array}
%$$
as we claimed. \qed
\end{proof}

Next lemma generalizes Cor. 40 of \citep{Ford1}. For the sake of completeness, we provide a direct complete proof of it.

\begin{lemma}\label{thm:Tomas1}
Let $P_{n}$ be a  shape invariant probabilistic model  of phylogenetic trees and let $P^*_{n}$ be the corresponding probabilistic model of  trees.
Then,  $P_{n}$ is sampling consistent if, and only if,  $P^*_{n}$ is sampling consistent.
\end{lemma}

\begin{proof}
Let us prove first the ``only if'' implication.
Let $P_{n}$ be sampling consistent. Then, for every $T_{n-1}^*\in \TT_{n-1}^*$ and for every $\widehat{T}_{n-1}\in\pi^{-1}(T_{n-1}^*)$,
$$
\begin{array}{l}
\displaystyle P_{n-1}^*(T_{n-1}^*)=\big|\{T_{n-1}\in \TT_{n-1}: \pi(T_{n-1})=T_{n-1}^*\}\big|\cdot P_{n-1}(\widehat{T}_{n-1})\\
\qquad \mbox{(by the shape invariance of $P_n$)}\\ 
\quad \displaystyle= \big|\{T_{n-1}\in \TT_{n-1}: \pi(T_{n-1})=T_{n-1}^*\}\big|\cdot\sum_{T_n\in \TT_n\atop T_n(-n)=\widehat{T}_{n-1}}\hspace*{-3ex} P_n(T_n)\\
\qquad \mbox{(by the sampling consistency of $P_n$)}\\
\quad \displaystyle=\sum_{T_{n-1}\in \pi^{-1}(T_{n-1}^*)}\sum_{T_n\in \TT_n\atop T_n(-n)={T}_{n-1}}\hspace*{-3ex} P_n(T_n)\\
\qquad \mbox{(by Lemma \ref{lem:prevT})}\\
\quad \displaystyle=\sum_{T_n\in \TT_n\atop \pi(T_n(-n))=T_{n-1}^*}\hspace*{-3ex} P_n(T_n)=\sum_{T_n\in \TT_n\atop \pi(T_n(-i))=T_{n-1}^*}\hspace*{-3ex} P_n(T_n)\quad\mbox{for every $i=1,\ldots,n$}\\
\qquad \mbox{(by the shape invariance of $P_n$).}\\
\end{array}
$$
Therefore
$$
\begin{array}{l}
\displaystyle n\cdot P_{n-1}^*(T_{n-1}^*)=\sum_{i=1}^n\sum_{T_n\in \TT_n\atop \pi(T_n(-i))=T_{n-1}^*}\hspace*{-3ex} P_n(T_n)\\
\quad \displaystyle =\sum_{T_n\in \TT_n} \big|\{i\in [n]: \pi(T_n(-i))=T_{n-1}^*\}\big|\cdot P_n(T_n)\\
\quad \displaystyle =\sum_{T^*_n\in \TT^*_n} \sum_{T_n\in \pi^{-1}(T_n^*)} \big|\{i\in [n]: \pi(T_n(-i))=T_{n-1}^*\}\big|\cdot P_n(T_n)\\
\quad \displaystyle =\sum_{T^*_n\in \TT^*_n}\Bigg(\big|\{x\in L(T_n^*): T_n^*(-x)=T_{n-1}^*\}\big| \cdot \sum_{T_n\in \pi^{-1}(T_n^*)}\hspace*{-2ex}P_n(T_n)\Bigg)\\
\quad \displaystyle =\sum_{T^*_n\in \TT^*_n} \big|\{x\in L(T_n^*): T_n^*(-x)=T_{n-1}^*\}\big|\cdot P_n^*(T_n^*)
\end{array}
$$
and hence
$$
P_{n-1}^*(T_{n-1}^*)=\sum_{T^*_n\in \TT^*_n} \frac{\big|\{x\in L(T_n^*): T_n^*(-x)=T_{n-1}^*\}\big|}{n}\cdot P_n^*(T_n^*)
$$
as we wanted to prove. 

The proof on the ``if'' implication consists in carefully running backwards the sequence of equalities in the proof of the ``only if'' implication. Indeed, assume that $P^*_{n}$ is sampling consistent and let $T_{n-1}\in \TT_{n-1}$ and $T^*_{n-1}=\pi(T_{n-1})\in \TT^*_{n-1}$. Then
$$
\begin{array}{l}
\displaystyle P_{n-1}^*(T_{n-1}^*)=\sum_{T^*_n\in \TT^*_n} \frac{\big|\{x\in L(T_n^*): T_n^*(-x)=T_{n-1}^*\}\big|}{n}\cdot P_n^*(T_n^*)\\
\qquad \mbox{(by the sampling consistency of $P^*_n$)}\\
\quad \displaystyle =\frac{1}{n}\sum_{T^*_n\in \TT^*_n}\Bigg(\big|\{x\in L(T_n^*): T_n^*(-x)=T_{n-1}^*\}\big| \cdot \sum_{T_n\in \pi^{-1}(T_n^*)}\hspace*{-2ex}P_n(T_n)\Bigg)\\
\quad \displaystyle =\frac{1}{n}\sum_{T^*_n\in \TT^*_n} \sum_{T_n\in \pi^{-1}(T_n^*)} \big|\{i\in [n]: \pi(T_n(-i))=T_{n-1}^*\}\big|\cdot P_n(T_n)\\
\quad \displaystyle =\frac{1}{n}\sum_{T_n\in \TT_n} \big|\{i\in [n]: \pi(T_n(-i))=T_{n-1}^*\}\big|\cdot P_n(T_n)\\
\quad \displaystyle = \frac{1}{n}\sum_{i=1}^n\sum_{T_n\in \TT_n\atop \pi(T_n(-i))=T_{n-1}^*}\hspace*{-3ex} P_n(T_n) =\sum_{T_n\in \TT_n\atop \pi(T_n(-n))=T_{n-1}^*}\hspace*{-3ex} P_n(T_n)\\
\qquad \mbox{(by the shape invariance of $P_n$)}\\
\quad \displaystyle=\sum_{T'_{n-1}\in \pi^{-1}(T_{n-1}^*)}\sum_{T_n\in \TT_n\atop T_n(-n)=T'_{n-1}}\hspace*{-3ex} P_n(T_n)\\
\quad \displaystyle= \big|\{T'_{n-1}\in \TT_{n-1}: \pi(T'_{n-1})=T_{n-1}^*\}\big|\cdot\sum_{T_n\in \TT_n\atop T_n(-n)=T_{n-1}}\hspace*{-3ex} P_n(T_n)\\
\qquad \mbox{(by Lemma \ref{lem:prevT})}
\end{array}
$$
and thus, dividing both sides of this equality by $\big|\{T'_{n-1}\in \TT_{n-1}: \pi(T'_{n-1})=T_{n-1}^*\}\big|$ and using the shape invariance of $P_n$, we obtain
$$
\sum_{T_n\in \TT_n\atop T_n(-n)=T_{n-1}}\hspace*{-3ex} P_n(T_n)=\frac{P_{n-1}^*(T_{n-1}^*)}{\big|\{T'_{n-1}\in \TT_{n-1}: \pi(T'_{n-1})=T_{n-1}^*\}\big|}=P_{n-1}(T_{n-1})
$$
as we wanted to prove.\qed
\end{proof}

{%
In Section \ref{sec:QI} we shall be concerned with three specific parametric probabilistic models  of phylogenetic trees: the $\beta$-model, the $\alpha$-model, and the $\alpha$-$\gamma$-model. To close this section, we provide detailed descriptions of these models and the explicit computation of the probabilities of all  trees with 4 leaves under them.

\subsubsection{Aldous' $\beta$-model.}
The $\beta$-splitting model  $P^A_{\beta,n}$ \citep{Ald1,Ald2}  is a probabilistic model of bifurcating phylogenetic trees that depends on one parameter $\beta\in (-2,\infty)$. Let us recall its definition.  For every $m\geq 2$ and $a=1,\ldots,m-1$, let 
$$
q_{m,\beta}(a)=\frac{1}{a_m(\beta)}\cdot \frac{\Gamma(\beta+a+1)\Gamma(\beta+m-a+1)}{\Gamma(a+1)\Gamma(m-a+1)},
$$
where $\Gamma$ stands for the usual Gamma function defined on $\RR^+$,
$$
\Gamma(x)=\int_0^\infty t^{x-1}e^{-t}\,dt,
$$
and $a_m(\beta)$ is a suitable normalizing constant so that $\sum\limits_{a=1}^{m-1} q_{m,\beta}(a)=1$. Recall (see, for instance, Chapter 6 in \citep{Gamma}) that $\Gamma$ satisfies that $\Gamma(x+1)=x\Gamma(x)$ and that, for every $n\in \NN$, 
$\Gamma(n+1)=n!$.

For every $m\geq 2$ and $a=1,\ldots,\lfloor m/2\rfloor$, let
$$
\widehat{q}_{m,\beta}(a)=\left\{\begin{array}{ll}
q_{m,\beta}(a)+q_{m,\beta}(m-a)=2q_{m,\beta}(a) & \mbox{ if $a\neq m/2$}\\
q_{m,\beta}(a) & \mbox{ if $a= m/2$}
\end{array}\right.
$$
With these notations, the probabilities under this model are computed as follows. Let $n\geq 1$  be a given desired number of leaves:
\begin{enumerate}
\item Start with a tree $T'_1$ consisting of a single node labeled $n$. Set $P'_{\beta,1}(T'_1)=1$.

\item At each step $j=1,\ldots,n-1$, the current tree $T_j'$ contains leaves with labels greater than 1. Then, choose equiprobably a leaf in $T_j'$ with a label  $m$  greater than $1$, choose a number $a=1,\ldots,\lfloor m/2\rfloor$ with probability distribution $\widehat{q}_{m,\beta}(a)$, and split this leaf into a cherry with a child labeled $a$ and a child labeled $m-a$. The resulting tree $T_{j+1}'$ has then probability
$$
P'_{\beta,j+1}(T'_{j+1})=\frac{\widehat{q}_{m,\beta}(a)}{|\{\mbox{\footnotesize leaves in $T_j'$ labeled $>1$}\}|}\cdot P'_{\beta,j}(T'_{j}).
$$

\item When the desired number $n$ of leaves is reached, all leaves are labeled 1 and $T_n'$ can be understood as a tree. Then, the probability of a given tree is defined as the sum of the probabilities of all ways of obtaining it by means of the previous procedure; that is, for every $T_n^*\in \BT_n^*$, its probability under the $\beta$-model is 
$$
P^{A,*}_{\beta,n}(T_n^*)=\sum_{T'_n = T_n^*} P'_{\beta,n}(T'_n).
$$

\item Finally, the probability $P^A_{\beta,n}(T)$ of any phylogenetic tree $T\in \BT_n$ is obtained from the probability under $P^{A,*}_{\beta,n}$ of its shape by means of equation (\ref{eq:pstar}):
$$
P^A_{\beta,n}(T)=\frac{P^{A,*}_{\beta,n}(\pi(T))}{\big|\{T'\in \BT_n: T'\equiv T\}\big|}.
$$
\end{enumerate}

The last step in the definition of  $P^A_{\beta,n}$ makes it shape invariant by construction, and \citet{Ald1} proves that it is sampling consistent. Hence, by Lemma \ref{thm:Tomas1}, the $\beta$-model of trees $P^{A,*}_{\beta,n}$ is also sampling consistent. This $\beta$-model  includes as specific cases the Yule model  \citep{Harding71,Yule} (when $\beta=0$) and the uniform model \citep{CS,Pinelis} (when $\beta=-3/2$).

In Section \ref{sec:QI} we shall need to know the probability $P^{A,*}_{\beta,4}$  of the maximally balanced tree with 4 leaves $((*,*),(*,*))$, which we denote in this paper by $Q_3^*$ (see Figure \ref{fig:5shapes} below). We compute this probability in the following lemma, taking the opportunity to provide a detailed example of how this model associates probabilities to trees through their construction.

\begin{lemma}\label{lem:AldB4}
For every $\beta\in (-2,\infty)$,
$$
P^{A,*}_{\beta,4}(Q_3^*)=\frac{3\beta+6}{7\beta+18}.
$$
\end{lemma}

\begin{proof}
We start with a single node labeled 4. In order to obtain a maximally balanced tree $((1,1),(1,1))$ using the previous procedure, in the first step we must split this node into a cherry with both leaves labeled 2. The probability of choosing this split is
$$
\widehat{q}_{4,\beta}(2)=q_{4,\beta}(2)=\frac{1}{a_4(\beta)}\cdot \frac{\Gamma(\beta+3)\Gamma(\beta+3)}{\Gamma(3)\Gamma(3)}.
$$
Let us compute the normalizing constant $a_4(\beta)$: since
$$
\begin{array}{l}
\displaystyle q_{4,\beta}(1)=q_{4,\beta}(3)=\frac{1}{a_4(\beta)}\cdot \frac{\Gamma(\beta+2)\Gamma(\beta+4)}{\Gamma(2)\Gamma(4)}\\[2ex]
\displaystyle q_{4,\beta}(2)=\frac{1}{a_4(\beta)}\cdot \frac{\Gamma(\beta+3)\Gamma(\beta+3)}{\Gamma(3)\Gamma(3)}
\end{array}
$$
imposing that $q_{4,\beta}(1)+q_{4,\beta}(2)+q_{4,\beta}(3)=1$ we obtain
$$
a_4(\beta)=\frac{2\Gamma(\beta+2)\Gamma(\beta+4)}{6}+\frac{\Gamma(\beta+3)^2}{4}=\frac{4\Gamma(\beta+2)\Gamma(\beta+4)+3\Gamma(\beta+3)^2}{12}.
$$
Therefore,
$$
q_{4,\beta}(2)=\frac{3\Gamma(\beta+3)^2}{4\Gamma(\beta+2)\Gamma(\beta+4)+3\Gamma(\beta+3)^2}.
$$

In the second step, we choose one of the leaves with probability $1/2$ and we split it into a cherry $(1,1)$. Since there is only one way of splitting a leaf labeled 2, $q_{2,\beta}(1)=1$. So, the probability of the tree obtained in this step is
$$
\frac{1}{2}q_{4,\beta}(2).
$$
Then, in the third step, we are forced to choose the other leaf labeled 2 and to split it  into a cherry $(1,1)$. We obtain a maximally balanced tree with all its leaves labeled 1 and its probability is still $q_{4,\beta}(2)/2$.

Now, there are two ways of obtaining the tree $((1,1),(1,1))$ with this construction, depending on which leaf of the cherry $(2,2)$ we choose to split first. So, the probability of the tree $Q_3^*$ is
$$
P^{A,*}_{\beta,4}(Q_3^*)=2\cdot \frac{1}{2}q_{4,\beta}(2)=\frac{3\Gamma(\beta+3)^2}{4\Gamma(\beta+2)\Gamma(\beta+4)+3\Gamma(\beta+3)^2}.
$$
Finally, using that $\Gamma(x+1)=x\Gamma(x)$, we have that
$$
\begin{array}{l}
\displaystyle \frac{3\Gamma(\beta+3)^2}{4\Gamma(\beta+2)\Gamma(\beta+4)+3\Gamma(\beta+3)^2}\\[2ex]
\qquad \displaystyle=
\frac{3(\beta+2)^2\Gamma(\beta+2)^2}{4(\beta+3)(\beta+2)\Gamma(\beta+2)^2+3(\beta+2)^2\Gamma(\beta+2)^2}=
\frac{3\beta+6}{7\beta+18}
\end{array}
$$
as we claimed.  
\qed
\end{proof}

\subsubsection{Ford's $\alpha$-model.}\label{sec:Ford} The $\alpha$-model $P^F_{\alpha,n}$ introduced by \citet{Ford1} is another probabilistic model of bifurcating phylogenetic trees that depends on one parameter $\alpha\in [0,1]$. It  is defined as follows. Let $n\geq 1$ be any desired number  of leaves:

\begin{enumerate}
\item Start with the tree $T_1$ consisting of a single node labeled 1. Set $P'_{\alpha,1}(T_1)=1$.

\item For every $m=1,\ldots,n-1$, let $T_{m+1}\in \BT_{m+1}$ be obtained by adding a new leaf labeled $m+1$ to $T_m$. Then:
\begin{itemize}
\item If  the new leaf is added to an arc  ending in  a leaf, 
$$
P'_{\alpha,m+1}(T_{m+1})=\frac{1 - \alpha}{m-\alpha}\cdot P'_{\alpha,m}(T_{m}).
$$

\item If  the new leaf is added to an arc   ending in  an internal node, 
$$
P'_{\alpha,m+1}(T_{m+1})=\frac{\alpha}{m-\alpha}\cdot P'_{\alpha,m}(T_{m}).
$$

\item If  the new leaf is added to a new root,
$$
P'_{\alpha,m+1}(T_{m+1})=\frac{\alpha}{m-\alpha}\cdot P'_{\alpha,m}(T_{m}).
$$
\end{itemize}

\item When the desired number $n$ of leaves is reached, the probability of a given tree is defined as the sum of the probabilities of all phylogenetic trees with that shape; that is, for every $T_n^*\in \BT_n^*$, its probability under the $\alpha$-model is 
$$
P^{F,*}_{\alpha,n}(T_n^*)=\sum_{\pi(T'_n) = T_n^*} P'_{\alpha,n}(T'_n).
$$

\item Finally, the probability $P^F_{\alpha,n}(T)$ of any phylogenetic tree $T\in \BT_n$ is obtained from the probability under $P^{F,*}_{\alpha,n}$ of its shape by means of equation (\ref{eq:pstar}):
$$
P^F_{\alpha,n}(T)=\frac{P^{F,*}_{\alpha,n}(\pi(T))}{\big|\{T'\in \BT_n: T'\equiv T\}\big|}.
$$
\end{enumerate} 
The $\alpha$-model is again shape invariant by construction and sampling consistent  by Prop. 42 of \citep{Ford1}, and it also includes as specific cases the Yule model  (when $\alpha=0$) and the uniform model (when $\alpha=1/2$).

In Section \ref{sec:QI}, we shall also need to know $P^{F,*}_{\alpha,4}(Q_3^*)$, where we recall that $Q_3^*$ stands for  the fully symmetric tree with 4 leaves.  This value was provided by \citet{Ford1} in Section 7, Fig. 20, as well as by \citet{MMR}. In the following lemma we compute it directly from the model's  definition to illustrate also in this case how the probability of a tree is obtained through its construction.

\begin{lemma}\label{lem:FordA4}
For every $\alpha\in [0,1]$,
$$
P_{\alpha,4}^{F,*}(Q_3^*)=\frac{1-\alpha}{3-\alpha}.
$$
\end{lemma}

\begin{proof}
To compute this probability, we shall already start with the cherry $T_2=(1,2)$ in $\BT_2$, which has probability $P'_{\alpha,2}(T_2)= 1$. Every tree in $\BT_3$ is obtained by adding a leaf labeled 3 to $T_2$. These trees are described in Figure \ref{fig:3btrees}.
Their probabilities are:
\begin{itemize}
\item $K^{(1)}$ and $K^{(2)}$ are obtained by adding the leaf 3 to an arc  in $T_2$ ending in a leaf. Their probability is then
$$
P'_{\alpha,3}(K^{(1)})=P'_{\alpha,3}(K^{(2)})=\frac{1 - \alpha}{2-\alpha}.
$$

\item $K^{(3)}$ is obtained by adding the leaf 3 to a new root. Its probability is then
$$
P'_{\alpha,3}(K^{(3)})=\frac{\alpha}{2-\alpha}.
$$
\end{itemize}

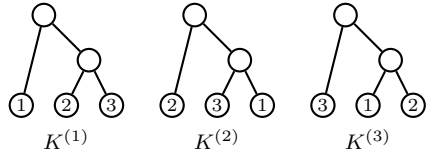
\begin{figure}[htb]
\begin{center}
\begin{tikzpicture}[thick,>=stealth,scale=0.3]
\draw(0,0) node [tre] (1)  {\scriptsize $1$};  
\draw(2,0) node [tre] (2)  {\scriptsize $2$};  
\draw(4,0) node [tre] (3)  {\scriptsize $3$};  
\draw(1,4) node [tre] (r) {};
\draw(3,2) node [tre] (a) {};  
\draw (r)--(a) ; 
\draw (r)--(1) ; 
\draw (a)--(2) ; 
\draw (a)--(3) ; 
\draw(2,-1.5) node  {\footnotesize $K^{(1)}$};  
\end{tikzpicture}
\quad
\begin{tikzpicture}[thick,>=stealth,scale=0.3]
\draw(0,0) node [tre] (1)  {\scriptsize $2$};  
\draw(2,0) node [tre] (2)  {\scriptsize $3$};  
\draw(4,0) node [tre] (3)  {\scriptsize $1$};  
\draw(1,4) node [tre] (r) {};
\draw(3,2) node [tre] (a) {};  
\draw (r)--(a) ; 
\draw (r)--(1) ; 
\draw (a)--(2) ; 
\draw (a)--(3) ; 
\draw(2,-1.5) node  {\footnotesize $K^{(2)}$};  
\end{tikzpicture}
\quad
\begin{tikzpicture}[thick,>=stealth,scale=0.3]
\draw(0,0) node [tre] (1)  {\scriptsize $3$};  
\draw(2,0) node [tre] (2)  {\scriptsize $1$};  
\draw(4,0) node [tre] (3)  {\scriptsize $2$};  
\draw(1,4) node [tre] (r) {};
\draw(3,2) node [tre] (a) {};  
\draw (r)--(a) ; 
\draw (r)--(1) ; 
\draw (a)--(2) ; 
\draw (a)--(3) ; 
\draw(2,-1.5) node  {\footnotesize $K^{(3)}$};  
\end{tikzpicture}
\end{center}
\caption{\label{fig:3btrees} The phylogenetic  trees in $\BT_3$.}
\end{figure}

Now, there are three phylogenetic trees in $\BT_4$ of shape  $Q_3^*$, depicted in Figure \ref{fig:4btrees}. Each one of them is obtained from the corresponding phylogenetic tree $K^{(i)}$ by adding the leaf 4 to the arc from the root to its only leaf child.  Their  probability is, then,
$$
P'_{\alpha,4}(Q_3^{(i)})=\frac{1 - \alpha}{3-\alpha}\cdot P'_{\alpha,3}(K^{(i)})
$$
and hence, since $\sum_{i=1}^3 P'_{\alpha,3}(K^{(i)})=1$,
$$
P^*_{\alpha,4}(Q_3^{*})=\sum_{i=1}^3  P'_{\alpha,4}(Q_3^{(i)})=\frac{1 - \alpha}{3-\alpha}
$$ 
as we claimed.   \qed
\end{proof}

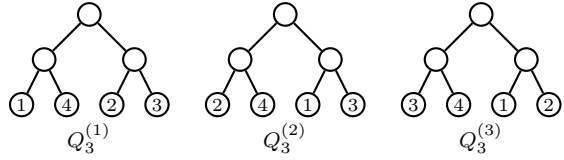
\begin{figure}[htb]
\begin{center}
\begin{tikzpicture}[thick,>=stealth,scale=0.3]
\draw(0,0) node [tre] (1)  {\scriptsize $1$};  
\draw(2,0) node [tre] (2)  {\scriptsize $4$};  
\draw(4,0) node [tre] (3)  {\scriptsize $2$};  
\draw(6,0) node [tre] (4)  {\scriptsize $3$};  
\draw(1,2) node [tre] (a) {};
\draw(5,2) node [tre] (b) {};  
\draw(3,4) node [tre] (r) {};  
\draw (r)--(a) ; 
\draw (r)--(b) ; 
\draw (a)--(1) ; 
\draw (a)--(2) ; 
\draw (b)--(3) ; 
\draw (b)--(4) ; 
\draw(3,-1.5) node  {\footnotesize $Q_3^{(1)}$};  
\end{tikzpicture}
\quad
\begin{tikzpicture}[thick,>=stealth,scale=0.3]
\draw(0,0) node [tre] (1)  {\scriptsize $2$};  
\draw(2,0) node [tre] (2)  {\scriptsize $4$};  
\draw(4,0) node [tre] (3)  {\scriptsize $1$};  
\draw(6,0) node [tre] (4)  {\scriptsize $3$};  
\draw(1,2) node [tre] (a) {};
\draw(5,2) node [tre] (b) {};  
\draw(3,4) node [tre] (r) {};  
\draw (r)--(a) ; 
\draw (r)--(b) ; 
\draw (a)--(1) ; 
\draw (a)--(2) ; 
\draw (b)--(3) ; 
\draw (b)--(4) ; 
\draw(3,-1.5) node  {\footnotesize $Q_3^{(2)}$};  
\end{tikzpicture}
\quad
\begin{tikzpicture}[thick,>=stealth,scale=0.3]
\draw(0,0) node [tre] (1)  {\scriptsize $3$};  
\draw(2,0) node [tre] (2)  {\scriptsize $4$};  
\draw(4,0) node [tre] (3)  {\scriptsize $1$};  
\draw(6,0) node [tre] (4)  {\scriptsize $2$};  
\draw(1,2) node [tre] (a) {};
\draw(5,2) node [tre] (b) {};  
\draw(3,4) node [tre] (r) {};  
\draw (r)--(a) ; 
\draw (r)--(b) ; 
\draw (a)--(1) ; 
\draw (a)--(2) ; 
\draw (b)--(3) ; 
\draw (b)--(4) ; 
\draw(3,-1.5) node  {\footnotesize $Q_3^{(3)}$};  
\end{tikzpicture}
\end{center}
\caption{\label{fig:4btrees} The fully symmetric phylogenetic  trees in $\BT_4$.}
\end{figure}

\subsubsection{Chen-Ford-Winkel's $\alpha$-$\gamma$-model.}
The $\alpha$-$\gamma$-model $P_{\alpha,\gamma,n}$, defined by \citet{Ford2}, is a probabilistic model of multifurcating phylogenetic trees that depends on two parameters $\alpha,\gamma$ with $0 \leq \gamma \leq \alpha \leq 1$.  It generalizes Ford's $\alpha$-model by allowing in the recursive construction of trees to add new leaves not only to arcs or to a new root, but also to internal nodes. More specifically,  the probability $P^{*}_{\alpha,\gamma,n}(T^*)$ of a tree $T^*\in \TT_n^*$ under this model is defined as follows. Let $n\geq 1$ be any desired number  of leaves:

\begin{enumerate}
\item Start with the tree $T_1\in \TT_1$ consisting of a single node labeled 1. Set $P_{\alpha,\gamma,1}(T_1)=1$.

\item For every $m=1,\ldots,n-1$, let $T_{m+1}\in \TT_{m+1}$ be obtained by adding a new leaf labeled $m+1$ to $T_m$. Then:
\begin{itemize}
\item If  the new leaf is added to an arc $e$ ending in a leaf, 
$$
P_{\alpha,\gamma,m+1}(T_{m+1})=\frac{1 - \alpha}{m-\alpha}\cdot P_{\alpha,\gamma,m}(T_{m}).
$$

\item If  the new leaf is added to an arc $e$ ending in an internal node, 
$$
P_{\alpha,\gamma,m+1}(T_{m+1})=\frac{\gamma}{m-\alpha}\cdot P_{\alpha,\gamma,m}(T_{m}).
$$

\item If  the new leaf is added to a new root,  
$$
P_{\alpha,\gamma,m+1}(T_{m+1})=\frac{\gamma}{m-\alpha}\cdot P_{\alpha,\gamma,m}(T_{m}).
$$

\item If  the new leaf is added as a child of an internal node $u$,  
$$
P_{\alpha,\gamma,m+1}(T_{m+1})=\frac{(\deg_{out}(u)-1)\alpha-\gamma}{m-\alpha}\cdot P_{\alpha,\gamma,m}(T_{m}).
$$
\end{itemize}

\item When the desired number $n$ of leaves is reached, the probability $P_{\alpha,\gamma,n}(T_n)$ of the resulting tree $T_n$ is the one obtained in this way. Then, the probability $P^{*}_{\alpha,\gamma,n}(T^*)$ of a given tree $T^*\in \TT^*_n$ is defined as the sum of the probabilities of all phylogenetic trees with that shape:
$$
P^{*}_{\alpha,\gamma,n}(T^*)=\sum_{\pi(T_n) = T^*} P_{\alpha,\gamma,n}(T_n).
$$
\end{enumerate} 
Notice that if $\alpha=\gamma$, this process only produces bifurcating trees and then, for every $T_n\in \BT_n$, $P_{\alpha,\alpha,n}(T_{n})=P'_{\alpha,n}(T_n)$ ---the provisional probability of $T_n$ defined by the recursive application of step 2 in the definition of the $\alpha$-model--- and, for every $T^*_n\in \BT^*_n$, $P^*_{\alpha,\alpha,n}(T^*_{n})=P^{F,*}_{\alpha,n}(T^*_n)$.

It turns out  that  $P_{\alpha,\gamma,n}$ is not shape invariant in general (see Prop. 1.(b) of \citep{Ford2}), but   the corresponding model for trees $P^*_{\alpha,\gamma,n}$  is  sampling consistent by Thm. 2 of \textsl{loc.\ cit}. 

 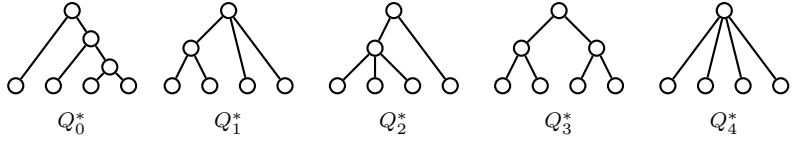
\begin{figure}[htb]
\begin{center}
\begin{tikzpicture}[thick,>=stealth,scale=0.25]
\draw(0,0) node [trep] (1) {};   
\draw(2,0) node [trep] (2) {};  
\draw(4,0) node [trep] (3) {};  
\draw(6,0) node [trep] (4) {};  
\draw(3,4) node[trep] (r) {};
\draw(4,2.5) node[trep] (a) {};
\draw(5,1) node[trep] (b) {};
\draw  (r)--(1);
\draw  (r)--(a);
\draw  (a)--(2);
\draw  (a)--(b);
\draw  (b)--(3);
\draw  (b)--(4);
\draw (3,-2) node {\footnotesize $Q_0^*$};
\end{tikzpicture}\quad
\begin{tikzpicture}[thick,>=stealth,scale=0.25]
\draw(0,0) node [trep] (1) {};   
\draw(2,0) node [trep] (2) {};  
\draw(4,0) node [trep] (3) {};  
\draw(6,0) node [trep] (4) {};  
\draw(1,2) node[trep] (a) {};
\draw(3,4) node[trep] (r) {};
\draw  (r)--(a);
\draw  (a)--(1);
\draw  (a)--(2);
\draw  (r)--(3);
\draw  (r)--(4);
\draw (3,-2) node {\footnotesize $Q_1^*$};
\end{tikzpicture}
\quad
\begin{tikzpicture}[thick,>=stealth,scale=0.25]
\draw(0,0) node [trep] (1) {};   
\draw(2,0) node [trep] (2) {};  
\draw(4,0) node [trep] (3) {};  
\draw(6,0) node [trep] (4) {};  
\draw(2,2) node[trep] (a) {};
\draw(3,4) node[trep] (r) {};
\draw  (r)--(a);
\draw  (a)--(1);
\draw  (a)--(2);
\draw  (a)--(3);
\draw  (r)--(4);
\draw (3,-2) node {\footnotesize $Q_2^*$};
\end{tikzpicture}
\quad
\begin{tikzpicture}[thick,>=stealth,scale=0.25]
\draw(0,0) node [trep] (1) {};   
\draw(2,0) node [trep] (2) {};  
\draw(4,0) node [trep] (3) {};  
\draw(6,0) node [trep] (4) {};  
\draw(1,2) node[trep] (a) {};
\draw(5,2) node[trep] (b) {};
\draw(3,4) node[trep] (r) {};
\draw  (r)--(a);
\draw  (a)--(1);
\draw  (a)--(2);
\draw  (r)--(b);
\draw  (b)--(3);
\draw  (b)--(4);
\draw (3,-2) node {\footnotesize $Q_3^*$};
\end{tikzpicture}
\quad
\begin{tikzpicture}[thick,>=stealth,scale=0.25]
\draw(0,0) node [trep] (1) {};   
\draw(2,0) node [trep] (2) {};  
\draw(4,0) node [trep] (3) {};  
\draw(6,0) node [trep] (4) {};  
\draw(3,4) node[trep] (r) {};
\draw  (r)--(1);
\draw  (r)--(2);
\draw  (r)--(3);
\draw  (r)--(4);
\draw (3,-2) node {\footnotesize $Q_4^*$};
\end{tikzpicture}
\end{center}
\caption{\label{fig:5shapes}The 5 trees in $\TT_4^*$.}
\end{figure}

Later in this paper we shall need to know the probabilities under $P^*_{\alpha,\gamma,4}$ of the five different trees in $\TT^*_4$, described in Figure \ref{fig:5shapes} together with the notations used in this paper to denote them (motivated by Table  \ref{tabla:indexq2} in the next section). We compute these probabilities in the following lemma, thus providing  an example of explicit computation of probabilities also for this model.

\begin{lemma}\label{lem:ag}
With the notations of Figure \ref{fig:5shapes}:
$$
\begin{array}{l}
\displaystyle 
P^*_{\alpha,\gamma,4}(Q_0^*)=\frac{2(1-\alpha+\gamma)(2(1 - \alpha)+\gamma)}{(3-\alpha)(2-\alpha)}
\\[2ex]
\displaystyle 
P^*_{\alpha,\gamma,4}(Q_1^*)=\frac{(5(1 - \alpha)+\gamma)(\alpha-\gamma)}{(3-\alpha)(2-\alpha)}\\[2ex]
\displaystyle 
P^*_{\alpha,\gamma,4}(Q_2^*)=\frac{2(1 - \alpha+\gamma)(\alpha-\gamma)}{(3-\alpha)(2-\alpha)}\\[2ex]
\displaystyle 
P^*_{\alpha,\gamma,4}(Q_3^*)=\frac{(1-\alpha)(2(1 - \alpha)+\gamma)}{(3-\alpha)(2-\alpha)}\\[2ex]
\displaystyle P^*_{\alpha,\gamma,4}(Q_4^*)=
\frac{(2\alpha-\gamma)(\alpha-\gamma)}{(3-\alpha)(2-\alpha)}\end{array}
$$
\end{lemma}

\begin{proof}
To compute these probabilities, we shall already start with the cherry $T_2=(1,2)$ in $\TT_2$, which has probability $P_{\alpha,\gamma,2}(T_2)= 1$. Every phylogenetic tree in $\TT_3$ is obtained by adding a leaf labeled 3 to $T_2$. There are 4 trees in $\TT_3$: the bifurcating trees $K^{(i)}$, $i=1,2,3$, described in Figure \ref{fig:3btrees}, and the rooted star $S_3$.
\begin{itemize}
\item $S_3$ is obtained by adding the leaf 3 to the root of $T_2$. Its probability is then
$$
P_{\alpha,\gamma,3}(S_3)=\frac{\alpha-\gamma}{2-\alpha}.
$$

\item $K^{(1)}$ and $K^{(2)}$ are obtained by adding the leaf 3 to an arc  in $T_2$ ending in a leaf. Their probability is then
$$
P_{\alpha,\gamma,3}(K^{(1)})=P_{\alpha,\gamma,3}(K^{(2)})=\frac{1 - \alpha}{2-\alpha}.
$$

\item $K^{(3)}$ is obtained by adding the leaf 3 to a new root. Its probability is then
$$
P_{\alpha,\gamma,3}(K^{(3)})=\frac{\gamma}{2-\alpha}.
$$
\end{itemize}

Let us move finally to $\TT_4^*$:
\begin{itemize}
\item A tree of shape $Q_4^*$ can only be obtained by adding the leaf 4 to the root of the tree $S_3$. Its probability is, then,
$$
P^*_{\alpha,\gamma,4}(Q_4^*)=\frac{2\alpha-\gamma}{3-\alpha}\cdot P_{\alpha,\gamma,3}(S_3)=
\frac{(2\alpha-\gamma)(\alpha-\gamma)}{(3-\alpha)(2-\alpha)}.
$$

\item A tree of shape  $Q_0^*$ can be obtained by adding the leaf 4 in some tree $K^{(i)}_3$ either to a new root, to the arc from the root to the other internal node,
or to one of the arcs in its cherry.  Its probability is, then,
$$
\begin{array}{rl}
P^*_{\alpha,\gamma,4}(Q_0^*) & \displaystyle=\Big(2\cdot\frac{\gamma}{3-\alpha}+2\cdot\frac{1 - \alpha}{3-\alpha}
\Big)\sum_{i=1}^3 P_{\alpha,\gamma,3}(K^{(i)})\\[1ex] & \displaystyle
=\frac{2(1-\alpha+\gamma)(2(1 - \alpha)+\gamma)}{(3-\alpha)(2-\alpha)}.
\end{array}
$$

\item A tree of shape $Q_1^*$ can be obtained by adding the leaf 4 either to one of the three arcs in the tree $S_3$ or to the root of some tree $K^{(i)}_3$.  Its probability is, then,
$$
\begin{array}{rl}
P^*_{\alpha,\gamma,4}(Q_1^*) & \displaystyle=3\cdot \frac{1 - \alpha}{3-\alpha}\cdot P_{\alpha,\gamma,3}(S_3)+
\frac{\alpha-\gamma}{3-\alpha} \sum_{i=1}^3 P_{\alpha,\gamma,3}(K^{(i)})\\[1ex] & \displaystyle
=\frac{(5(1 - \alpha)+\gamma)(\alpha-\gamma)}{(3-\alpha)(2-\alpha)}.
\end{array}
$$

\item A tree of shape  $Q_2^*$ can be obtained by adding the leaf 4 either to a new root in the tree $S_3$ or to the non-root internal node in some tree $K^{(i)}_3$.  Its probability is, then,
$$
\begin{array}{rl}
P^*_{\alpha,\gamma,4}(Q_2^*) & \displaystyle=\frac{\gamma}{3-\alpha}\cdot P_{\alpha,\gamma,3}(S_3)+
\frac{\alpha-\gamma}{3-\alpha}\sum_{i=1}^3 P_{\alpha,\gamma,3}(K^{(i)})\\[1ex] & \displaystyle
=\frac{2(1 - \alpha+\gamma)(\alpha-\gamma)}{(3-\alpha)(2-\alpha)}.
\end{array}
$$

\item A tree of shape  $Q_3^*$ can only be obtained by adding the leaf 4 to the arc from the root to its only leaf child in some tree $K^{(i)}_3$.  Its probability is, then,
$$
P^*_{\alpha,\gamma,4}(Q_3^*)=\frac{1 - \alpha}{3-\alpha}\sum_{i=1}^3 P_{\alpha,\gamma,3}(K^{(i)})
=\frac{(1-\alpha)(2(1 - \alpha)+\gamma)}{(3-\alpha)(2-\alpha)}.
$$
\end{itemize}
\hspace*{\fill}  \qed
\end{proof}

Notice that, when $\alpha=\gamma$,
$$
P^*_{\alpha,\alpha,4}(Q_3^*)=\frac{1-\alpha}{3-\alpha}=P^{F,*}_{\alpha,4}(Q_3^*)
$$
as it should have been expected.

}

\section{Rooted quartet indices}

Let $T$ be a phylogenetic tree on  a set $\Sigma$. For every $Q\in \mathcal{P}_4(\Sigma)$,  the \emph{rooted quartet on $Q$ displayed} by  $T$ is the restriction $T(Q)$ of $T$ to $Q$. 
A phylogenetic tree $T\in \TT_n$ can contain rooted quartets of five different shapes, namely, those listed in Figure \ref{fig:5shapes}. Notice that a bifurcating phylogenetic tree $T\in \BT_n$ can only contain rooted quartets of two shapes: those denoted by $Q_0^*$ and $Q_3^*$ in the aforementioned figure.

We associate to each rooted quartet an \emph{$\mathit{rQI}$-value} $q_i$ that increases with the symmetry of the rooted quartet's shape, as measured by means of its number of  automorphisms, going from a value $q_0=0$ for  the least symmetric tree, the comb $Q_0^*$, to a largest value of $q_4$ for the most symmetric one, the rooted star $Q_4^*$; see Table \ref{tabla:indexq2}. The specific numerical values can be chosen in order to magnify the differences in symmetry between specific pairs of trees. For instance, one could take  $q_i=i$, or $q_i=2^i$.

\begin{table}[htb]
\begin{center}
\begin{tabular}{|l|c|c|c|c|c|}
\hline
Rooted quartet\vphantom{$\displaystyle \Big($} & $Q_0^*$ & $Q_1^*$&  $Q_2^*$ & $Q_3^*$ & $Q_4^*$      \\ \hline 
\# Automorphisms\vphantom{$\displaystyle \Big($} & 2 & 4 & 6 & 8 & 24  \\ \hline 
$\mathit{rQI}$\vphantom{$\displaystyle \Big($} & $0$ & $q_1$ & $q_2$ & $q_3$ & $q_4$
\\ \hline
\end{tabular}
\end{center}
\caption{\label{tabla:indexq2}The rooted quartets' $\mathit{rQI}$-values, with $0<q_1<q_2<q_3<q_4$.}
\end{table}

Now, for every $T\in \TT(\Sigma)$, we define its \emph{rooted quartet index} $\mathit{rQI}(T)$ as the sum of the $\mathit{rQI}$-values of its rooted quartets:
$$
\begin{array}{rl}
\mathit{rQI}(T) &\displaystyle =\sum_{Q\in \mathcal{P}_4(\Sigma)} \mathit{rQI}(T(Q))\\
& \displaystyle =\sum_{i=1}^4 \big |\{Q\in \mathcal{P}_4(\Sigma): \pi(T(Q))=Q_i^*\}\big|\cdot q_i
\end{array}
$$
In particular, if $|\Sigma|\leq 3$, then $\mathit{rQI}(T)=0$ for every $T\in \TT(\Sigma)$. So, \emph{we shall assume henceforth that $|\Sigma|\geq 4$}.

It is clear that $\mathit{rQI}$ is a \emph{shape index}, in the sense that two phylogenetic trees with the same shape have the same rooted quartet index. It makes sense then to define the rooted quartet index $\mathit{rQI}(T^*)$ of a tree $T^*\in \TT_n^*$ as the rooted quartet index of any phylogenetic tree of shape $T^*$.

\begin{example}
Consider the tree $T=((1,2,3),4,(5,(6,7)))$ depicted in Figure \ref{fig:exqi}. It has:  4 rooted quartets of shape $Q_0^*$;  18 rooted quartets of shape $Q_1^*$;  4 rooted quartets of shape $Q_2^*$;  9 rooted quartets of shape $Q_3^*$; and no rooted quartet of shape $Q_4^*$.
Therefore
$$
\mathit{rQI}(T)=18q_1+4q_2+9q_3.
$$
\end{example}

 \begin{figure}[htb]
\begin{center}
\begin{tikzpicture}[thick,>=stealth,scale=0.3]
\draw(0,0) node [tre] (1) {}; \etq 1  
\draw(2,0) node [tre] (2) {};  \etq 2
\draw(4,0) node [tre] (3) {};  \etq 3
\draw(6,0) node [tre] (4) {};  \etq 4
\draw(8,0) node [tre] (5) {};  \etq 5
\draw(10,0) node [tre] (6) {};  \etq 6
\draw(12,0) node [tre] (7) {};  \etq 7
\draw (2,3) node [tre] (a) {};
\draw (11,2) node [tre] (b) {};
\draw (9,4) node [tre] (c) {};
\draw (6,6) node [tre] (r) {};
\draw  (r)--(a);
\draw  (r)--(c);
\draw  (r)--(4);
\draw  (a)--(1);
\draw  (a)--(2);
\draw  (a)--(3);
\draw  (c)--(5);
\draw  (c)--(b);
\draw  (b)--(6);
\draw  (b)--(7);
\end{tikzpicture}
\end{center}
\caption{\label{fig:exqi}The tree $((1,2,3),4,(5,(6,7)))$.}
\end{figure}
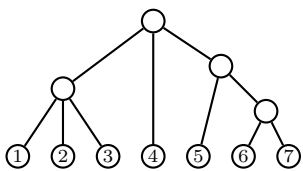

\begin{remark}
If we did not take $q_0=0$, then the resulting index would be
$$
\begin{array}{l}
\displaystyle\sum_{i=0}^4 q_i\cdot \big|\{Q\in \mathcal{P}_4(\Sigma): \pi(T(Q))=Q_i^*\}\big|\\
\qquad\qquad \displaystyle =q_0\binom{n}{4}+ \sum_{i=1}^4 (q_i-q_0) \big|\{Q\in \mathcal{P}_4(\Sigma): \pi(T(Q))=Q_i^*\}\big|\Big)
\end{array}
$$
which is equivalent (up to the constant addend $q_0\binom{n}{4}$) to $\mathit{rQI}$ taking as  $\mathit{rQI}$-values $q_i'=q_i-q_0$.
\end{remark}

\begin{remark}
One could also associate other values to the rooted quartet shapes; for instance their  Sackin index \citep{Sackin:72,Shao:90}  or their total cophenetic index \citep{MRR}, which measure the imbalance of the rooted quartet's shape, 
from a smallest value at $Q_4^*$ to a largest value at $Q_0^*$. All results obtained in this paper  are easily translated to any other sets of values.
\end{remark}

Since a bifurcating tree can only contain rooted quartets of shape $Q_0^*$ and  $Q_3^*$,  its $\mathit{rQI}$ index is simply $q_3$ times its number of rooted quartets of shape $Q_3^*$. Therefore, in order to avoid this spurious factor, when dealing only with bifurcating trees we shall use the following alternative \emph{rooted quartet index for bifurcating trees} $\mathit{rQIB}$: for every $T\in \BT(\Sigma)$,
$$
\mathit{rQIB}(T)=\frac{1}{q_3}\mathit{rQI}(T)=\Big|\big\{Q\in \mathcal{P}_4(\Sigma):  \pi(T(Q))= Q_3^*\big\}\Big|.
$$
%In particular, $\mathit{rQIB}(K_n)=\mathit{rQI}(K_n)= 0$.

The rooted quartet index for bifurcating trees satisfies the following recurrence.

\begin{lemma}\label{lem:rec}
Let $T=T_1\star T_2\in \BT_n$, where each $T_i$ has $n_i$ leaves. Then,
$$
\mathit{rQIB}(T)= \mathit{rQIB}(T_1)+ \mathit{rQIB}(T_2)+\binom{n_1}{2}\cdot \binom{n_2}{2}.
$$
\end{lemma}

\begin{proof}
For every $Q\in \mathcal{P}_4([n])$,  there are the following possibilities:
\begin{enumerate}[(1)]
\item If $Q\subseteq L(T_i)$, for some $i=1,2$,  then $T(Q)=T_i(Q)$. Therefore, 
each $Q\subseteq L(T_i)$ contributes $\mathit{rQIB}(T_i)$ to  $\mathit{rQIB}(T)$.
\item If three leaves in $Q$ belong to one of the subtrees $T_i$ and the fourth to the other subtree $T_j$, then $T(Q)$ has shape $Q_0^*$ and thus it does not contribute anything to $\mathit{rQIB}(T)$.
\item If two leaves in $Q$ belong to $T_1$ and the other two to $T_2$, then $T(Q)$ has shape $Q_3^*$ and thus it contributes 1 to $\mathit{rQIB}(T)$. There are $\binom{n_1}{2}\cdot \binom{n_2}{2}$ such quartets of leaves $Q$.
\end{enumerate}
Adding up all these contributions, we obtain the formula in the statement.\qed
\end{proof}

Thus, $\mathit{rQIB}$   is a  \emph{recursive tree shape statistic} in the sense of \citet{Matsen}. The recurrence in the last lemma implies directly the following explicit formula for $\mathit{rQIB}$, which in particular entails that it can be easily computed in time $O(n)$, with $n$ the number of leaves of the tree, by traversing the tree in post-order (cf. the first paragraph in the proof of Proposition \ref{prop:costQI} below):

\begin{corollary}
If, for every $T\in \BT_n$ and for every $v\in V_{int}(T)$, we set $\mathrm{child}(v)=\{v_1,v_2\}$, then 
$$
\mathit{rQIB}(T)=\sum_{v\in V_{int}(T)} \binom{\kappa(v_1)}{2}\cdot\binom{\kappa(v_2)}{2}.
$$
\end{corollary}

Unfortunately, $\mathit{rQI}$ is not recursive in this sense:
there does not exist any family of mappings $q_m: \NN^m\to \RR$, $m\geq 2$, such that, for every $T\in \TT_n$, if $T=T_1\star \cdots \star T_m$, with each $T_i$ having $n_i$ leaves, then
$$
\mathit{rQI}(T)=\sum_{i=1}^m \mathit{rQI}(T_i)+q_m(n_1,\ldots,n_m).
$$
However, next lemma shows that there exists a slightly more involved linear  recurrence for $\mathit{rQI}$, with its independent term depending on more indices of the trees $T_i$ than only their numbers of leaves,  which still allows its computation in linear time.

For every $T\in \TT_n$, let $\Upsilon(T)$ be the number of \emph{non-bifurcating triples} in $T$  (that is, of restrictions of  $T$ to sets of 3 leaves  that have the shape of a rooted star $S_3$; cf. Fig.~\ref{fig:exs}). Notice that if $T=T_1\star\cdots\star T_m$ and $|L(T_i)|=n_i$, for each $i=1,\ldots,m$, then
$$
\Upsilon(T)=\sum_{i=1}^m \Upsilon(T_i)+\sum_{1\leq i_1<i_2<i_3\leq m}n_{i_1}n_{i_2}n_{i_3}
$$
and hence 
$$
\Upsilon(T)=\sum_{v\in V_{int}(T)}\sum_{\{v_1,v_2,v_3\}\subseteq \mathrm{child}(v)} \kappa(v_1)\kappa(v_2)\kappa(v_3).
$$

\begin{lemma}
Let $T=T_1\star \cdots \star T_m\in \TT_n$, where each $T_i$ has $n_i$ leaves. Then
$$
\begin{array}{l}
\mathit{rQI}(T)\displaystyle =\sum_{i=1}^m \mathit{rQI}(T_i) +q_4\cdot\sum_{1\leq i_1<i_2<i_3<i_4\leq m} n_{i_1}n_{i_2}n_{i_3}n_{i_4}  \\[1ex]
\qquad \displaystyle+q_3\cdot\sum_{1\leq i_1<i_2\leq m} \binom{n_{i_1}}{2}\binom{n_{i_2}}{2}
+q_2\cdot\sum_{1\leq i_1<i_2\leq m}\hspace*{-1ex} \big(n_{i_1}\Upsilon(T_{i_2})+n_{i_2}\Upsilon(T_{i_1})\big)\\[1ex]
\qquad \displaystyle +q_1\cdot\sum_{1\leq i_1<i_2<i_3\leq m}\Big(\binom{n_{i_1}}{2}n_{i_2}n_{i_3}+\binom{n_{i_2}}{2}n_{i_1}n_{i_3}+\binom{n_{i_3}}{2}n_{i_1}n_{i_2}\Big).
\end{array}
$$
\end{lemma}

\begin{proof}
For every $Q\in \mathcal{P}_4([n])$,  there are the following possibilities:
\begin{enumerate}[(1)]
\item If $Q\subseteq L(T_i)$, for some $i$,  then $T(Q)=T_i(Q)$. Therefore, 
each $Q\subseteq L(T_i)$ contributes $\mathit{rQI}(T_i)$ to  $\mathit{rQI}(T)$.

\item If 3 leaves, say $a,b,c$, in $Q$ belong to a subtree $T_i$ and the fourth to another subtree  $T_j$, then $T(Q)$:
\begin{itemize}
\item Has shape $Q_2^*$ if $T_i(\{a,b,c\})$ has shape $S_3$. For every pair of subtrees $T_i,T_j$, there are $n_j\Upsilon(T_i)+n_i\Upsilon(T_j)$ quartets of leaves $Q$ of this type, and each one of them contributes $q_2$ to $\mathit{rQI}(T)$
\item Has shape $Q_0^*$ if $T_i(\{a,b,c\})$ is a comb $K_3$. These rooted quartets do not contribute anything to $\mathit{rQI}(T)$.
\end{itemize}

\item If 2 leaves in $Q$ belong to a subtree $T_i$ and the other 2 to another subtree  $T_j$, then $T(Q)$ has shape $Q_3^*$. For every pair of subtrees $T_i,T_j$, there are $\binom{n_i}{2}\binom{n_j}{2}$  quartets of leaves $Q$ of this type, and each one of them contributes $q_3$ to $\mathit{rQI}(T)$.

\item If 2 leaves in $Q$ belong  to a subtree $T_i$, a third leaf to another subtree  $T_j$ and the fourth   to a third subtree  $T_k$,  then $T(Q)$ has shape $Q_1^*$. For every triple of subtrees $T_i,T_j,T_k$, there are $\binom{n_i}{2}n_jn_k+\binom{n_j}{2}n_kn_i+\binom{n_k}{2}n_in_j$ quartets of leaves $Q$ of this type, and each one of them contributes $q_1$ to $\mathit{rQI}(T)$.

\item If each leaf in $Q$ belongs to a different subtree $T_i$, then $T(Q)$ has shape $Q_4^*$. 
For every four subtrees $T_i,T_j,T_k,T_l$, there are $n_in_jn_kn_l$ such quartets of leaves $Q$, and each one of them contributes $q_4$ to $\mathit{rQI}(T)$.
\end{enumerate}
Adding up all these contributions, we obtain the formula in the statement.\qed
\end{proof}

\begin{proposition}\label{prop:costQI}
If $T\in \TT_n$, $\mathit{rQI}(T)$ can be computed in time $O(n)$.
\end{proposition}

\begin{proof}
Let $T$ be a phylogenetic tree in $\TT_n$. Recall that if a certain mapping $\phi: V(T)\to \RR$ can be computed in constant time at each leaf of $T$ and
 in $O(\deg(v))$ time at each internal node $v$ from its value at the children of $v$, then the whole vector $(\phi(v))_{v\in V(T)}$, and hence also its sum $\sum\limits_{v\in V(T)} \phi(v)$, can be computed in $O(n)$ time by traversing $T$ in post-order. Indeed, if we denote by $m_k$ the number of internal nodes of $T$ with out-degree $k$, then the cost of computing $(\phi(v))_{v\in V(T)}$ through a post-order traversal of $T$ is  $O\big(n+\sum_k m_k\cdot k\big)$, and $\sum_k m_k\cdot k$ is the number of arcs in $T$, which is at most $2n-2$. We shall use this remark several times in this proof, and, to begin with, we refer to it to recall that the vector $\big(\kappa(v)\big)_{v\in V(T)}$ can be computed in $O(n)$ time.

Now, in order to simplify the notations,  let, for every $v\in V_{int}(T)$:
$$
\begin{array}{rl}
E_l(v)& \displaystyle =\hspace*{-0.4cm} \sum_{\{v_1,\ldots,v_l\}\subseteq \mathrm{child}(v)} \kappa(v_1)\cdots \kappa(v_l),\qquad l=2,\ldots,\deg(v)\\
\displaystyle F_1(v)& \displaystyle =\hspace*{-0.7cm} \sum_{\{v_1,v_2,v_3\}\subseteq \mathrm{child}(v)} \Big(\binom{\kappa(v_1)}{2}\kappa(v_2)\kappa(v_3)+
\binom{\kappa(v_2)}{2}\kappa(v_1)\kappa(v_3)\\
& \displaystyle\hspace*{2.5cm}+\binom{\kappa(v_3)}{2}\kappa(v_1)\kappa(v_2)\Big)\\
%& \displaystyle =\frac{1}{2}\sum_{\{v_1,v_2,v_3\}\subseteq \mathrm{child}(v)} \kappa(v_1)\kappa(v_2)\kappa(v_3)(\kappa(v_1)+\kappa(v_2)+\kappa(v_3)-3)\\
\displaystyle F_2(v)& \displaystyle =\hspace*{-0.4cm}\sum_{\{v_1,v_2\}\subseteq \mathrm{child}(v)} \big(\kappa(v_1)\Upsilon(T_{v_2})+\kappa(v_2)\Upsilon(T_{v_1})\big)\\
\displaystyle F_3(v)& \displaystyle =\hspace*{-0.4cm} \sum_{\{v_1,v_2\}\subseteq \mathrm{child}(v)} \binom{\kappa(v_1)}{2}\binom{\kappa(v_2)}{2}\\
\end{array}
$$
so that
$$
\begin{array}{l}
\displaystyle \Upsilon(T)=\sum_{v\in V_{int}(T)} E_3(v)\\
\displaystyle \mathit{rQI}(T)=\sum_{v\in V_{int}(T)}\big(q_1F_1(v)+q_2 F_2(v)+q_3F_3(v)+q_4E_4(v))
\end{array}
$$
We want to prove now that each one of the vectors
$$
\big(F_1(v)\big)_{v\in V_{int}(T)},\
\big(F_2(v)\big)_{v\in V_{int}(T)},\
\big(F_3(v)\big)_{v\in V_{int}(T)},
\big(E_4(v)\big)_{v\in V_{int}(T)}
$$
can be computed in  $O(n)$ time, which will clearly entail that $\mathit{rQI}(T)$ can be computed in  $O(n)$ time.

One of the key ingredients in the proof are the \emph{Newton-Girard formulas} (see, for instance, Section I.2 in \citep{McDonald79}): given a (multi)set of numbers $X=\{x_1,\ldots,x_k\}$, if we let, for every $l\geq 1$,
$$
P_l(X)=\sum_{i=1}^k x_i^l,\quad E_l(X)=\sum_{1\leq i_1<\cdots< i_l\leq k} x_{i_1}\cdots x_{i_l}
$$
then
$$
E_l(X)=\frac{1}{l!}\left|
\begin{array}{cccccc}
P_1(X) & 1 & 0 & \ldots & 0 & 0\\
P_2(X) & P_1(X) & 2 &\ldots & 0 & 0\\
P_3(X) & P_2(X) & P_1(X) & \ldots & 0 & 0\\
\vdots & \vdots & \vdots & \ddots & \vdots& \vdots\\
P_{l-1}(X) & P_{l-2}(X) & P_{l-3}(X) & \ldots & P_1(X) & l-1\\
P_l(X) & P_{l-1}(X) & P_{l-2}(X) & \ldots & P_2(X) & P_1(X)
\end{array}
\right|
$$
If we consider $l$ as a fixed parameter, every $P_l(X)$ can be computed in time $O(k)$ and then this expression for $E_l(X)$ as an $l\times l$ determinant allows us also to compute it  in time  $O(k)$.

In particular, if, for every $v\in V_{int}(V)$, we consider the multiset $X_v=\{\kappa(u): u\in \mathrm{child}(v)\}$, then every $E_l(v)=E_l(X_v)$ can be computed in time $O(\deg(v))$ and hence the whole vector $\big(E_l(v)\big)_{v\in V_{int}(T)}$ can be computed in time $O(n)$. In particular,  $\big(E_3(v)\big)_{v\in V_{int}(T)}$ and $\big(E_4(v)\big)_{v\in V_{int}(T)}$ can be computed in linear time.

Then, using the recursion 
$$
\Upsilon(T_v)=\sum_{v_i\in\mathrm{child}(v)} \Upsilon(T_{v_i})+E_3(v)
$$
we deduce that the whole vector $\big(\Upsilon(T_v)\big)_{v\in V_{int}(T)}$ can also be computed in time $O(n)$. Now,
$$
\begin{array}{rl}
F_2(v)& \displaystyle=\hspace*{-0.4cm}\sum_{\{v_1,v_2\}\subseteq \mathrm{child}(v)} \big(\kappa(v_1)\Upsilon(T_{v_2})+\kappa(v_2)\Upsilon(T_{v_1})\big)\\
& \displaystyle=
\Big(\sum_{v_i\in \mathrm{child}(v)} \kappa(v_i)\Big)\Big(\sum_{v_j\in \mathrm{child}(v)} \Upsilon(T_{v_j})\Big)-
\sum_{v_i\in \mathrm{child}(v)} \kappa(v_i)\Upsilon(T_{v_i})\\
& \displaystyle= \kappa(v)\big(\Upsilon(T_v)-E_3(v))-
\sum_{v_i\in \mathrm{child}(v)} \kappa(v_i)\Upsilon(T_{v_i}),
\end{array}
$$
This implies that each $F_2(v)$ can be computed in time $O(\deg(v))$ and hence that the whole vector 
$\big(F_2(v)\big)_{v\in V_{int}(T)}$ can be computed in time $O(n)$.

Let us focus now on
$$
\begin{array}{rl}
F_3(v) & \displaystyle = \sum_{\{v_1,v_2\}\subseteq \mathrm{child}(v)} \binom{\kappa(v_1)}{2}\binom{\kappa(v_2)}{2}\\
& \displaystyle= \frac{1}{4}\sum_{\{v_1,v_2\}\subseteq \mathrm{child}(v)} \kappa(v_1)^2\kappa(v_2)^2+\frac{1}{4}\sum_{\{v_1,v_2\}\subseteq \mathrm{child}(v)}\kappa(v_1)\kappa(v_2)\\
& \displaystyle \qquad\qquad -\frac{1}{4}\sum_{\{v_1,v_2\}\subseteq \mathrm{child}(v)}\big(\kappa(v_1)^2\kappa(v_2)+\kappa(v_2)^2\kappa(v_1)\big)
\end{array}
$$
In this expression, 
$$
\sum_{\{v_1,v_2\}\subseteq \mathrm{child}(v)}\kappa(v_1)\kappa(v_2)=E_2(v),\ \sum_{\{v_1,v_2\}\subseteq \mathrm{child}(v)}\kappa(v_1)^2\kappa(v_2)^2=E_2(X_v^2),
$$
where $X^2_v=\{\kappa(u)^2: u\in \mathrm{child}(v)\}$,
and hence they are computed in time $O(\deg(v))$. 
As far as the subtrahend goes,
$$
\begin{array}{l}
\displaystyle\sum_{\{v_1,v_2\}\subseteq \mathrm{child}(v)}\big(\kappa(v_1)^2\kappa(v_2)+\kappa(v_2)^2\kappa(v_1)\big)\\
\displaystyle\qquad\qquad =
\Big(\sum_{v_i\in \mathrm{child}(v)} \kappa(v_i)^2\Big)\Big(\sum_{v_j\in \mathrm{child}(v)} \kappa(v_j)\Big)-
\Big(\sum_{v_i\in \mathrm{child}(v)} \kappa(v_i)^3\Big)
\end{array}
$$
and hence it can also be computed in time $O(\deg(v))$. Therefore, the whole vector 
$\big(F_3(v)\big)_{v\in V_{int}(T)}$ can be computed in time $O(n)$.

Let us consider finally $F_1(v)$. We have that
$$
\begin{array}{rl}
\displaystyle F_1(v)& \displaystyle =\hspace*{-0.7cm} \sum_{\{v_1,v_2,v_3\}\subseteq \mathrm{child}(v)} \Big(\binom{\kappa(v_1)}{2}\kappa(v_2)\kappa(v_3)\\
& \displaystyle \hspace*{2cm}+
\binom{\kappa(v_2)}{2}\kappa(v_1)\kappa(v_3)+\binom{\kappa(v_3)}{2}\kappa(v_1)\kappa(v_2)\Big)\\
& \displaystyle =\frac{1}{2}\sum_{\{v_1,v_2,v_3\}\subseteq \mathrm{child}(v)} \kappa(v_1)\kappa(v_2)\kappa(v_3)(\kappa(v_1)+\kappa(v_2)+\kappa(v_3)-3)\\
& \displaystyle =\frac{1}{2}\sum_{\{v_1,v_2,v_3\}\subseteq \mathrm{child}(v)} \big(\kappa(v_1)^2\kappa(v_2)\kappa(v_3)\\
& \displaystyle\hspace*{2cm}+\kappa(v_2)^2\kappa(v_1)\kappa(v_3)+\kappa(v_3)^2\kappa(v_1)\kappa(v_2)\big)
-\frac{3}{2}E_3(v)\\
& \displaystyle =\frac{1}{2}\Big(\sum_{\{v_1,v_2,v_3\}\subseteq \mathrm{child}(v)} \kappa(v_1)\kappa(v_2)\kappa(v_3)\Big)\Big(\sum_{v_i\in \mathrm{child}(v)} \kappa(v_i)\Big)\\
\\
& \displaystyle\qquad\qquad-2\Big(\sum_{\{v_1,v_2,v_3.v_4\}\subseteq \mathrm{child}(v)} \kappa(v_1)\kappa(v_2)\kappa(v_3)\kappa(v_4)\Big)-\frac{3}{2}E_3(v)\\
& \displaystyle=\frac{1}{2}E_3(v)E_1(v)-2E_4(v)-\frac{3}{2}E_3(v)
\end{array}
$$
This expression shows that $F_1(v)$ can  be computed in time $O(\deg(v))$ and therefore the whole vector 
$\big(F_1(v)\big)_{v\in V_{int}(T)}$ can be computed in time $O(n)$. \qed
\end{proof}

\section{Trees with maximum and minimum $\mathit{rQI}$}\label{sec:maximin}

Let $n\geq 4$. In this section we determine which  trees in $\TT_n$ and $\BT_n$ have the largest and smallest corresponding rooted quartet indices. 
The multifurcating case is easy:

\begin{theorem}
The minimum value of $\mathit{rQI}$ in $\TT_n$ is reached exactly at the combs $K_n$, and it is 0. The maximum value
value of $\mathit{rQI}$ in $\TT_n$ is reached exactly at the rooted star $S_n$, and it is $q_4\binom{n}{4}$.
\end{theorem}

\begin{proof}
Since the $\mathit{rQI}$-value of a rooted quartet goes from 0 to $q_4$, we have that $0\leq \mathit{rQI}(T)\leq q_4 \binom{n}{4}$,
for every $T\in \TT_n$. Now, all rooted quartets displayed by a comb $K_n$ have shape $Q_0^*$, and therefore $\mathit{rQI}(K_n)=0$, while all rooted quartets displayed by $S_n$ have shape $Q_4^*$, and therefore $\mathit{rQI}(S_n)=q_4 \binom{n}{4}$.

As far as the uniqueness of the trees yielding the maximum and minimum values of $\mathit{rQI}$ goes, notice that, on the one hand, if $T$ is not a comb, then it displays some rooted quartet of shape other than   $Q_0^*$, because it  contains either some internal node of out-degree greater than 2, which becomes the root of some multifurcating rooted quartet, or  two cherries that determine a rooted quartet of shape $Q_3^*$. This implies that if $T\neq K_n$, then $\mathit{rQI}(T)>0$.
On the other hand, if $T\neq S_n$, then its root has some child that is not a leaf and therefore
$T$ displays some rooted quartet of shape other than   $Q_4^*$, which implies that $\mathit{rQI}(T)<q_4 \binom{n}{4}$. 
\qed
\end{proof}

Therefore, the range of $\mathit{rQI}$ on $\TT_n$ goes from 0 to $q_4\binom{n}{4}$. This is one order of magnitude wider than the range of the total cophenetic index  \citep{MRR}, which, going from 0 to $\binom{n}{3}$, was so far the balance index in the literature with the widest range.

We shall now characterize  those \emph{bifurcating} phylogenetic trees with largest $\mathit{rQI}$, or, equivalently, with largest $\mathit{rQIB}$.  
They turn out to be exactly the maximally balanced trees, as defined at the end of Subsection 2.1. The proof is similar to that of the characterization of the bifurcating phylogenetic trees with minimum total cophenetic index provided in Section 4 of \citep{MRR}.

\begin{lemma}\label{lem:min}
Let $T\in \BT_n$ be the bifurcating phylogenetic tree depicted in Fig \ref{fig:min1}.(a).
For every  $i=1,2,3,4$, let 
$n_i=|L(T_i)|$, and assume that $n_1> n_3$ and $n_2>n_4$. Then, $\mathit{rQIB}(T)$ is not maximum in $\BT_n$.
\end{lemma}

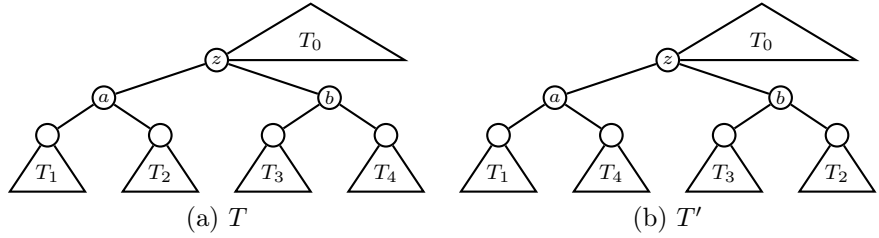
\begin{figure}[htb]
\begin{center}
\begin{tikzpicture}[thick,>=stealth,scale=0.25]
\draw(0,2) node[tre] (v1) {};   
\draw(6,2) node[tre] (v2) {}; 
\draw(12,2) node[tre] (v3) {}; 
\draw(18,2) node[tre] (v4) {}; 
\draw(3,4) node[tre] (a) {}; \etq a
\draw(15,4) node[tre] (b) {}; \etq b
\draw(9,6) node[tre] (z) {}; \etq z
\draw (z)--(19,6)--(14,9)--(z);
\draw(14,7) node  {\footnotesize $T_0$};
\draw (z)--(a);
\draw (z)--(b);
\draw (a)--(v1);
\draw (a)--(v2);
\draw (b)--(v3);
\draw (b)--(v4);
\draw (v1)--(-2,-1)--(2,-1)--(v1);
\draw(0,0) node  {\footnotesize $T_1$};
\draw (v2)--(4,-1)--(8,-1)--(v2);
\draw(6,0) node  {\footnotesize $T_2$};
\draw (v3)--(10,-1)--(14,-1)--(v3);
\draw(12,0) node  {\footnotesize $T_3$};
\draw (v4)--(16,-1)--(20,-1)--(v4);
\draw(18,0) node  {\footnotesize $T_4$};
\draw(9,-2.5) node {(a) $T$};
\end{tikzpicture}
\quad
\begin{tikzpicture}[thick,>=stealth,scale=0.25]
\draw(0,2) node[tre] (v1) {};   
\draw(6,2) node[tre] (v2) {}; 
\draw(12,2) node[tre] (v3) {}; 
\draw(18,2) node[tre] (v4) {}; 
\draw(3,4) node[tre] (a) {}; \etq a
\draw(15,4) node[tre] (b) {}; \etq b
\draw(9,6) node[tre] (z) {}; \etq z
\draw (z)--(19,6)--(14,9)--(z);
\draw(14,7) node  {\footnotesize $T_0$};
\draw (z)--(a);
\draw (z)--(b);
\draw (a)--(v1);
\draw (a)--(v2);
\draw (b)--(v3);
\draw (b)--(v4);
\draw (v1)--(-2,-1)--(2,-1)--(v1);
\draw(0,0) node  {\footnotesize $T_1$};
\draw (v2)--(4,-1)--(8,-1)--(v2);
\draw(6,0) node  {\footnotesize $T_4$};
\draw (v3)--(10,-1)--(14,-1)--(v3);
\draw(12,0) node  {\footnotesize $T_3$};
\draw (v4)--(16,-1)--(20,-1)--(v4);
\draw(18,0) node  {\footnotesize $T_2$};
\draw(9,-2.5) node {(b) $T'$};
\end{tikzpicture}
\end{center}
\caption{\label{fig:min1} 
(a) The tree $T$ in the statement of Lemma \ref{lem:min}.
(b) The tree $T'$ in the proof of Lemma \ref{lem:min}.}
\end{figure}

\begin{proof}
Let $T'$  be the tree obtained from $T$ by interchanging $T_2$ and $T_4$; see Fig \ref{fig:min1}.(b). We shall prove that $\mathit{rQIB}(T')>\mathit{rQIB}(T)$.

Let $\Sigma_z$ be the set of labels of $T_z$, which is also the set of labels of $T_z'$.
 To simplify the language, we shall understand the common subtree $T_0$ of $T$ and $T'$ as a phylogenetic tree on $([n]\setminus \Sigma_z)\cup\{z\}$.
Then, for every $Q=\{a,b,c,d\}\in \mathcal{P}_4([n])$:
\begin{itemize}
\item If $Q\cap \Sigma_z=\emptyset$, then $T(Q)=T'(Q)=T_0(Q)$.
\item If $Q\cap  \Sigma_z$ is a single label, say $d$, then $T(Q)=T'(Q)=T_0(\{a,b,c,z\})$.
\item If $Q\cap  \Sigma_z$ consists of two labels, say $c,d$, then $T(Q)=T'(Q)$. More specifically: $T(Q)=T'(Q)=((a,b),(c,d))$ when $T_0(\{a,b,z\})=((a,b),z)$; $T(Q)=T'(Q)=(a, (b,(c,d)))$ when $T_0(\{a,b,z\})=(a,(b,z))$; and $T(Q)=T'(Q)=(b, (a,(c,d)))$ when $T_0(\{a,b,z\})=(b,(a,z))$.
\item If $Q\cap  \Sigma_z$ consists of three labels,  then $T(Q)$ and $T'(Q)$ are both combs.
\end{itemize}
Therefore, $T(Q)$ and $T'(Q)$ can only be different when $Q\subseteq  \Sigma_z$, in which case
 $T(Q)=T_z(Q)$ and $T'(Q)=T'_z(Q)$. This implies that
$$
\mathit{rQIB}(T')-\mathit{rQIB}(T)=\mathit{rQIB}(T_z')-\mathit{rQIB}(T_z).
$$
Now, to compute the difference in the right hand side of this equality,
we apply Lemma \ref{lem:rec}:
$$
\begin{array}{rl}
\mathit{rQIB}(T_z) & \displaystyle =
\mathit{rQIB}(T_1)+\mathit{rQIB}(T_2) +\mathit{rQIB}(T_3)+\mathit{rQIB}(T_4)\\
& \displaystyle\qquad+\binom{n_1}{2}\binom{n_2}{2}+\binom{n_3}{2}\binom{n_4}{2}+ \binom{n_1+n_2}{2}\binom{n_3+n_4}{2}\\[2ex]
 \mathit{rQIB}(T'_z) & \displaystyle=
\mathit{rQIB}(T_1)+\mathit{rQIB}(T_4)+\mathit{rQIB}(T_2) +\mathit{rQIB}(T_3)\\
& \displaystyle\qquad+\binom{n_1}{2}\binom{n_4}{2}+\binom{n_2}{2}\binom{n_3}{2}+ \binom{n_1+n_4}{2}\binom{n_2+n_3}{2}
\end{array}
$$
and hence
$$
\mathit{rQIB}(T'_z)-\mathit{rQIB}(T_z)=\frac{1}{2}(n_1-n_3)(n_2-n_4)(n_1n_3+n_2n_4)>0
$$
% \displaystyle= \binom{n_1}{2}\binom{n_4}{2}+\binom{n_2}{2}\binom{n_3}{2}+ \binom{n_1+n_4}{2}\binom{n_2+n_3}{2}\\ \qquad\qquad \displaystyle 
%-
%\binom{n_1}{2}\binom{n_2}{2}-\binom{n_3}{2}\binom{n_4}{2}- \binom{n_1+n_2}{2}\binom{n_3+n_4}{2}\\[2ex]
%%\qquad \displaystyle= \binom{n_1}{2}\binom{n_4}{2}+\binom{n_2}{2}\binom{n_3}{2}+ \Bigg(\binom{n_1}{2}+\binom{n_4}{2}+n_1n_4\Bigg)\Bigg(\binom{n_2}{2}+\binom{n_3}{2}+n_2n_3\Bigg) \\ \qquad\qquad \displaystyle 
%%-
%%\binom{n_1}{2}\binom{n_2}{2}-\binom{n_3}{2}\binom{n_4}{2}- \Bigg(\binom{n_1}{2}+\binom{n_2}{2}+n_1n_2\Bigg)\Bigg(\binom{n_3}{2}+\binom{n_4}{2}+n_3n_4\Bigg)\\
%%\qquad\displaystyle = (n_2-n_4)\Bigg(n_3\binom{n_1}{2}-n_1\binom{n_3}{2}\Bigg)+(n_1-n_3)\Bigg(n_4\binom{n_2}{2}-n_2\binom{n_4}{2}\Bigg)\\
%\qquad\displaystyle = \frac{1}{2}(n_1-n_3)(n_2-n_4)(n_1n_3+n_2n_4)>0
%\end{array}
because $n_1> n_3$ and $n_2>n_4$ by assumption.  \qed
\end{proof}

\begin{lemma}\label{lem:min2}
Let $T\in \BT_n$ be a bifurcating phylogenetic tree containing a leaf whose sibling has at least 3 descendant leaves.
Then, $\mathit{rQIB}(T)$ is not maximum in $\BT_n$.
\end{lemma}

\begin{figure}[htb]
\begin{center}
\begin{tikzpicture}[thick,>=stealth,scale=0.25]
\draw(0,2) node[tre] (v1) {};   
\draw(6,2) node[tre] (v2) {}; 
\draw(3,4) node[tre] (a) {}; \etq a
\draw(15,-0.5) node[tre] (m) {\scriptsize $\ell$};
\draw(9,6) node[tre] (z) {}; \etq z
\draw (z)--(16,6)--(12,9)--(z);
\draw(12,7) node  {\footnotesize $T_0$};
\draw (z)--(a);
\draw (z)--(m);
\draw (a)--(v1);
\draw (a)--(v2);
\draw (v1)--(-2,-1)--(2,-1)--(v1);
\draw(0,0) node  {\footnotesize $T_1$};
\draw (v2)--(4,-1)--(8,-1)--(v2);
\draw(6,0) node  {\footnotesize $T_2$};
\draw(7.5,-2.5) node {(a) $T$};
\end{tikzpicture}
\quad
\begin{tikzpicture}[thick,>=stealth,scale=0.25]
\draw(0,2) node[tre] (v1) {};   
\draw (v1)--(-2,-1)--(2,-1)--(v1);
\draw(0,0) node  {\footnotesize $T_1$};
\draw(8,6) node[tre] (z) {}; \etq z
\draw (z)--(16,6)--(12,9)--(z);
\draw(12,7) node  {\footnotesize $T_0$};
\draw (z)--(v1);
\draw(6,2) node[tre] (v2) {}; 
\draw (v2)--(4,-1)--(8,-1)--(v2);
\draw(6,0) node  {\footnotesize $T_2$};
\draw(11,4) node[tre] (b) {}; \etq b
\draw(16,-0.5) node[tre] (m)  {\scriptsize $\ell$};
\draw (z)--(b);
\draw (b)--(m);
\draw (b)--(v2);
\draw(7.5,-2.5) node {(b) $T'$};
\end{tikzpicture}
\end{center}
\caption{\label{fig:min3} 
(a) The tree $T$ in the statement of Lemma \ref{lem:min2}. (b) The tree $T'$ in the proof of Lemma \ref{lem:min2}.}
\end{figure}
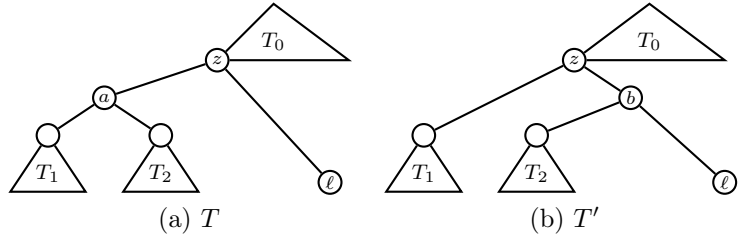

\begin{proof}
Assume that the tree $T\in \BT_n$ in the statement is the one depicted in Fig \ref{fig:min3}.(a), with $\ell$ a leaf such that the subtree $T_a$ rooted at its sibling $a$ has $|L(T_a)|\geq 3$. Let 
$n_1=|L(T_1)|$ and $n_2=|L(T_2)|$ and assume $n_1\geq n_2$: then, since $n_1+n_2\geq 3$,  $n_1\geq 2$. Let then $T'$ be the tree depicted in Fig \ref{fig:min3}.(b): we shall prove that $\mathit{rQIB}(T')>\mathit{rQIB}(T)$.
Reasoning as in the proof of the last lemma, we deduce that
$$
\mathit{rQIB}(T')-\mathit{rQIB}(T)=\mathit{rQIB}(T_z')-\mathit{rQIB}(T_z).
$$
Now, using Lemma \ref{lem:rec}, we have that
$$
\begin{array}{rl}
\mathit{rQIB}(T_z) & \displaystyle = \mathit{rQIB}(T_a)=\mathit{rQIB}(T_1)+\mathit{rQIB}(T_2)+\binom{n_1}{2}\cdot \binom{n_2}{2}\\
\mathit{rQIB}(T'_z) & \displaystyle = \mathit{rQIB}(T_1)+\mathit{rQIB}(T_b)+\binom{n_1}{2}\cdot \binom{n_2+1}{2}\\
& \displaystyle=
\mathit{rQIB}(T_1)+\mathit{rQIB}(T_2)+\binom{n_1}{2}\cdot \binom{n_2+1}{2}
\end{array}
$$
and therefore 
$$
\mathit{rQIB}(T'_z)-\mathit{rQIB}(T_z)=n_2\binom{n_1}{2}>0
$$
as we wanted to prove.  \qed
\end{proof}

\begin{theorem}\label{th:min}
For every $T\in \BT_n$, $\mathit{rQIB}(T)$ is maximum in $\BT_n$ if, and only if, $T$ is maximally balanced.
\end{theorem}

\begin{proof}
Assume that   $\mathit{rQIB}(T)$ is maximum in $\BT_n$ but that $T\in \BT_n$ is not maximally balanced, and let us reach a contradiction. Let $z$ be a non-balanced internal node in $T$ such that all its proper descendant internal nodes are balanced, and let  $a$ and $b$ be its children, with $\kappa(a)\geq \kappa(b)+2$.

If $b$ is a leaf, then, by Lemma \ref{lem:min2}, $\mathit{rQIB}(T)$ cannot be maximum in $\BT_n$. Therefore, $a$ and $b$ are internal, and hence balanced. Let $v_1,v_2$ be the children of $a$, $v_3,v_4$ the children of $b$, and $n_i=\kappa(v_i)$, for $i=1,2,3,4$. Without any loss of generality, we shall assume that $n_1\geq n_2$ and $n_3\geq n_4$. Then, since $a$ and $b$ are balanced, $n_1=n_2$ or $n_2+1$ and $n_3=n_4$ or $n_4+1$. Then, $n_1+n_2=\kappa(a)\geq \kappa(b)+2=n_3+n_4+2$ implies that  $n_1>n_3$.

Now, by Lemma \ref{lem:min}, since by assumption $\mathit{rQIB}(T)$ is maximum on $\BT_n$, it must happen that 
$n_1>n_3\geq n_4\geq n_2$. This forbids the equality $n_1=n_2$, and hence $n_1-1=n_2=n_3=n_4$. But this contradicts 
that $n_1+n_2\geq n_3+n_4+2$.

This implies that a non maximally balanced tree in $\BT_n$ cannot have maximum $\mathit{rQIB}$, and therefore
the maximum $\mathit{rQIB}$ in $\BT_n$ is reached at the maximally balanced trees, which have all the same shape and hence the same $\mathit{rQIB}$ index.
  \qed
\end{proof}

So, the only bifurcating trees with maximum $\mathit{rQIB}$ (and hence with maximum $\mathit{rQI}$) are the maximally balanced.
This maximum value of $\mathit{rQIB}$ on $\BT_n$ is given by the following recurrence. 

\begin{lemma}\label{lem:recmin}
For every $n$, let $b_n$ be the maximum of $\mathit{rQIB}$ on $\BT_n$. Then, $b_1=b_2=b_3=0$ and
$$
b_n=b_{\lceil n/2\rceil}+b_{\lfloor n/2\rfloor}+\binom{\lceil n/2\rceil}{2}\cdot\binom{\lfloor n/2\rfloor}{2},\quad \mbox{ for }n\geq 4.
$$
\end{lemma}

\begin{proof}
This recurrence for $b_n$ is a direct consequence of Lemma \ref{lem:rec} and the fact that the root of a maximally balanced tree in $\BT_n$ is balanced and the subtrees rooted at their children are maximally balanced.   \qed
\end{proof}

The sequence $b_n$ seems to be new, in the sense that it has no relation with any sequence previously contained  in Sloane's \textsl{On-Line Encyclopedia of Integer Sequences}  \citep{Sloane}.  Its values for $n=4,\ldots,20$ are
$$
1,3,9,19,   38,   64,  106 , 162,  243,  343,  479,  645,  860, 1110, 1424,
1790, 2237.
$$
%q=c(0,0,0,1)
%for (i in (4:20)){
%q[i]=q[floor(i/2)]+q[ceiling(i/2)]+choose(floor(i/2),2)*+choose(ceiling(i/2),2)
%}
%q
It is easy to prove, using the Master theorem for solving recurrences \citep[Thm. 4.1]{Cormen}, that $b_n$ grows asymptotically in $O(n^4)$. Moreover, it is easy to compute $b_{2^n}$ from this recurrence, yielding
$$
b_{2^n} = \Big(\frac{4}{7 (2^n - 3)} + \frac{3}{7}\Big)\binom{2^{n}}{4}.
$$
In particular, $b_{2^n}/\binom{2^{n}}{4}\stackrel{n\to\infty}{\longrightarrow} 3/7$, which is in agreement with the probability of the fully symmetric rooted quartet $Q_3^*$ under the $\beta$-model when $\beta\to\infty$; cf. Section 4.1 in \citep{Ald1}.

%
%\begin{remark}
%Some properties of $(b_n)_n$, for what they are worth:
%\begin{itemize}
%\item $b_{n+1}-b_{n}=b_{\lfloor n/2\rfloor+1}- b_{\lfloor n/2\rfloor}+\lfloor n/2\rfloor\cdot \binom{\lceil n/2\rceil}{2}$
%\item $b_n=\sum_{k=0}^{n-1} a_k$ where $a_k=a_{\lfloor k/2\rfloor}+\lfloor k/2\rfloor\cdot \binom{\lceil k/2\rceil}{2}$
%\end{itemize}
%They may help us to better understand it.
%\end{remark}

{\begin{remark}
When the range of values of a shape index $I$ grows with the number of leaves $n$  of the phylogenetic trees, as it is the case with $\mathit{rQI}$ and $\mathit{rQIB}$, it makes no sense to compare directly its value on two trees with different numbers of leaves. To overcome this drawback, one usually \emph{normalizes} the index, so that its range becomes independent on $n$. A suitable way to do that is to use the generic affine transformation
$$
\widetilde{I}(T)=\frac{I(T)-\min I(\TT_n)}{\max I(\TT_n)-\min I(\TT_n)}
$$
where $n$ stands for the number of leaves of the tree $T$ and $I(\TT_n)$ denotes the set of values of $I$ on $\TT_n$.  In this way, for every number of leaves, the minimum value of the normalized index is always 0 and the maximum value is always 1.

 As to our $\mathit{rQI}$, its minimum value is always 0, but its maximum depends on whether we are considering multifurcating or bifurcating trees. Therefore, we propose two normalized versions of this index:
\begin{itemize}
\item On $\TT_n$, $\widetilde{\mathit{rQI}}(T)=\mathit{rQI}(T)/(q_4\binom{n}{4})$.

\item On $\BT_n$, $\widetilde{\mathit{rQIB}}(T)=\mathit{rQIB}(T)/b_n$, with $b_n$ computed by means of the recurrence given in Lemma \ref{lem:recmin}.
\end{itemize}
\end{remark}}

\section{The expected value and the variance of $\mathit{rQI}$}\label{sec:QI}

Let $P_n$ be a probabilistic model of phylogenetic trees
 and  $\mathit{rQI}_n$  the random variable that chooses a  phylogenetic tree $T\in \TT_n$ with probability distribution $P_n$ and computes $\mathit{rQI}(T)$. 
 In this section we are interested in obtaining expressions for the expected value $E_{P}(\mathit{rQI}_n)$ and the variance $Var_{P}(\mathit{rQI}_n)$
 of $\mathit{rQI}_n$ under suitable models $P_n$.
 
Next lemma shows that, to compute these values, we can restrict ourselves to work with unlabeled trees. 
Let $P^*_n$ the probabilistic model of  trees induced by $P_n$ and $\mathit{rQI}^*_n$  the random variable that chooses a tree $T^*\in \TT^*_n$ with probability distribution $P^*_n$ and computes $\mathit{rQI}(T^*)$, defined as $\mathit{rQI}(T)$ for some phylogenetic tree $T$ of shape $T^*$.

\begin{lemma}\label{lem:igualesp}
For every $n\geq 1$, the distributions of $\mathit{rQI}_n$ and $\mathit{rQI}_n^*$ are the same. In particular,
their expected values and their variances  are the same. 
\end{lemma}

\begin{proof}
Let $f_{\mathit{rQI}_n}$ and $f_{\mathit{rQI}_n^*}$ be the probability density functions of the discrete random variables $\mathit{rQI}_n$ and $\mathit{rQI}_n^*$, respectively. Then, for every $x_0\in \RR$,
$$
\begin{array}{rl}
f_{\mathit{rQI}_n}(x_0) & \displaystyle =\hspace*{-2ex}\sum_{T\in \TT_n\atop \mathit{rQI}(T)=x_0}\hspace*{-1ex} P_n(T)=\hspace*{-2ex}\sum_{T^*\in \TT^*_n\atop \mathit{rQI}(T^*)=x_0} \sum_{T\in \TT_n\atop \pi(T)=T^*}\hspace*{-1ex} P_n(T)\\
& \displaystyle =\hspace*{-2ex}
\sum_{T^*\in \TT^*_n\atop \mathit{rQI}(T^*)=x_0}\hspace*{-1ex} P^*_n(T^*)=f_{\mathit{rQI}_n^*}(x_0)
\end{array}
$$
\ \hspace*{\fill}\qed
\end{proof}

\begin{proposition}\label{thm:EQI}
If  $P_n^*$ is sampling consistent, then
$$
E_{P}(\mathit{rQI}_n)=\binom{n}{4}\sum_{i=1}^4 q_i P^*_4(Q_i^*).
$$
\end{proposition}

\begin{proof}
By Lemma \ref{lem:igualesp}, $E_{P}(\mathit{rQI}_n)$ is equal to the expected value $E_{P^*}(\mathit{rQI}^*_n)$ of  $\mathit{rQI}^*_n$ under $P_n^*$, which can be computed as follows:
$$
\begin{array}{l}
E_{P^*}(\mathit{rQI}^*_n)\displaystyle = \sum_{T^*\in \TT^*_n} \mathit{rQI}(T^*)P^*_n(T^*)
\\[2ex]
\qquad \displaystyle =\sum_{T^*\in \TT^*_n} \Big(\sum_{i=1}^4 q_i\big|\{Q\in\mathcal{P}_4(L(T^*)): T^*(Q)=Q_i^*\}\big|\Big)P^*_n(T^*)
\\[2ex]
\qquad \displaystyle =\binom{n}{4}\sum_{i=1}^4 q_i\hspace*{-1.5ex}\sum_{T^*\in \TT^*_n} \frac{\big|\{Q\in\mathcal{P}_4(L(T^*)): T^*(Q)=Q_i^*\}\big|}{\binom{n}{4}} P^*_n(T^*)\\[2ex]
\qquad \displaystyle =\binom{n}{4}\sum_{i=1}^4 q_i P^*_4(Q_i^*)
\end{array}
$$
because, for every $i=1,\ldots,4$,
$$
\sum_{T^*\in \TT^*_n} \frac{\big|\{Q\in\mathcal{P}_4(L(T^*)): T^*(Q)=Q_i^*\}\big|}{\binom{n}{4}} P^*_n(T^*)=P^*_4(Q_i^*)
$$
by the sampling consistency of $P_n^*$.
  \qed
\end{proof}

This expression for $E_{P}(\mathit{rQI}_n)$ should not be surprising: by the sampling consistency property,  for each $i=1,\ldots 4$, the expected number of rooted quartets of shape $Q_i^*$ in a tree of $n$ leaves is $\binom{n}{4}P^*_4(Q_i^*)$ and their weight in $\mathit{rQI}$ value is $q_i$.

The $\alpha$-$\gamma$-model  for unlabeled trees $P_{\alpha,\gamma,n}^*$ is sampling consistent \citep[Prop. 12]{Ford2}. Therefore, applying the last proposition using the distribution $P^*_{\alpha,\gamma,4}$ on $\TT^*_4$ given in  Lemma \ref{lem:ag}, we have the following result.

\begin{corollary}
Let  $P_{\alpha,\gamma,n}$  be the $\alpha$-$\gamma$-model of phylogenetic trees, with $0\leq\gamma\leq \alpha\leq 1$.
Then
$$
\begin{array}{rl}
E_{P_{\alpha,\gamma}}(\mathit{rQI}_n) & \displaystyle=\binom{n}{4}\Bigg(
\frac{(2\alpha - \gamma)(\alpha - \gamma)}{(3 - \alpha)(2 - \alpha)}\cdot q_4+
\frac{(1-\alpha)(2(1-\alpha)+\gamma)}{(3 - \alpha)(2 - \alpha)}\cdot q_3\\[1ex]
 & \displaystyle\qquad +
 \frac{2(1-\alpha+\gamma)(\alpha-\gamma)}{(3 - \alpha)(2 - \alpha)}\cdot q_2+
\frac{(5(1 - \alpha)+\gamma)(\alpha-\gamma)}{(3 - \alpha)(2 - \alpha)}\cdot q_1\Bigg).
\end{array}
$$
\end{corollary}

If $P_n$ is a probabilistic model of bifurcating phylogenetic trees, so that 
$P^*_4(Q_1^*)=P^*_4(Q_2^*)=P^*_4(Q_4^*)=0$, then the expression in Prop. \ref{thm:EQI} becomes
$$
E_{P}(\mathit{rQI}_n)=\binom{n}{4}q_3P^*_4(Q_3^*).
$$
Taking $q_3=1$, we obtain the following results.

\begin{corollary}\label{prop:expbin}
If $P_n$ is a probabilistic model of bifurcating phylogenetic trees such that  $P_n^*$ is sampling consistent, then
$$
E_{P}(\mathit{rQIB}_n)=\binom{n}{4}P^*_4(Q_3^*).
$$
\end{corollary}

Since the $\beta$ and $\alpha$-models of bifurcating  (unlabeled) trees are sampling consistent, this corollary together with 
the probabilities of $Q_3^*\in \BT^*_4$ under these models given in Lemmas   \ref{lem:AldB4} and \ref{lem:FordA4}, respectively, entail the following result.

\begin{corollary}
Let  $P^A_{\beta,n}$  be Aldous' $\beta$-model for bifurcating phylogenetic trees, with $\beta\in (-2,\infty)$, and let $P^F_{\alpha,n}$  be Ford's $\alpha$-model for bifurcating phylogenetic trees, with $\alpha\in [0,1]$. 
Then:
$$
E_{P^A_{\beta}}(\mathit{rQIB}_n) \displaystyle=\frac{3\beta+6}{7\beta+18}\binom{n}{4},\quad E_{P^F_{\alpha}}(\mathit{rQIB}_n)=\frac{1-\alpha}{3-\alpha}\binom{n}{4}.
$$
\end{corollary}

It is easy to check that $E_{P^F_{\alpha}}(\mathit{rQIB}_n)$ agrees with $E_{P_{\alpha,\gamma}}(\mathit{rQI}_n)$ (up to the factor $q_3$) when $\alpha=\gamma$.

In particular, under the Yule model, which corresponds to $\alpha=0$ or $\beta=0$, and the uniform model, which corresponds to $\alpha=1/2$ or $\beta=-3/2$, the expected values of $\mathit{rQIB}_n$ are, respectively,
$$
E_{Y}(\mathit{rQIB}_n)=\frac{1}{3}\binom{n}{4},\quad E_{U}(\mathit{rQIB}_n) =\frac{1}{5}\binom{n}{4}.
$$

Let us deal now with the variance of $\mathit{rQI}_n$. To simplify the notations, for every $k=5,6,7,8$, for every $T^*\in \TT_k^*$ and for every $i,j\in \{1,2,3,4\}$, let
$$
\begin{array}{rl}
\Theta_{i,j}(T^*) & =\big|\{(Q,Q')\in\mathcal{P}_4(L(T^*))^2: \\
& \hspace*{1cm} Q\cup Q'=L(T^*), T^*(Q)=Q_i^*, T^*(Q')=Q_j^*\}\big|\\
& =\big|\{(Q,Q')\in\mathcal{P}_4(L(T^*))^2: \\
& \hspace*{1cm} |Q\cap Q'|=8-k, T^*(Q)=Q_i^*, T^*(Q')=Q_j^*\}\big|.
\end{array}
$$
Notice that $\Theta_{i,j}(T^*)=\Theta_{j,i}(T^*)$.

\begin{proposition}\label{prop:varqi}
If  $P_n^*$ is sampling consistent, then
$$
\begin{array}{l}
\displaystyle Var_{P}(\mathit{rQI}_n) =
\binom{n}{4}\sum_{i=1}^4 q_i^2P_4^*(Q_i^*)- \binom{n}{4}^2\Bigg(\sum_{i=1}^4q_iP^*_4(Q_i^*)\Bigg)^2\\
\displaystyle \qquad\quad+\sum_{i=1}^4\sum_{j=1}^4q_iq_j \Bigg(\sum_{k=5}^8\binom{n}{k} \sum_{T^*\in \TT_{k}^*}\Theta_{i,j}(T^*)P^*_{k}(T^*)\Bigg).
\end{array}
$$
\end{proposition}

\begin{proof}
Since, by Lemma \ref{lem:igualesp}, $Var_{P}(\mathit{rQI}_n)=Var_{P^*}(\mathit{rQI}^*_n)$, we shall compute the latter using the  formula $Var_{P^*}(\mathit{rQI}^*_n)=E_{P^*}({\mathit{rQI}^*_n}^2)-E_{P^*}({\mathit{rQI}^*_n})^2$, and therefore we need to compute  $E_{P^*}({\mathit{rQI}^*_n}^2)$. 

For every $T^*\in \TT^*_n$, for every $Q_i^*\in \TT^*_4$ and for every $Q\in \mathcal{P}_4(L(T^*))$, set
$$
\delta(Q;Q_i^*;T^*)=\left\{\begin{array}{ll}
1 & \mbox{ if $T^*(Q)=Q_i^*$}\\
0  & \mbox{ if $T^*(Q)\neq  Q_i^*$}
\end{array}\right.
$$
Then:
$$
\begin{array}{l}
E_{P^*}({\mathit{rQI}^*_n}^2)  \displaystyle = \sum_{T^*\in \TT^*_n} \mathit{rQI}^*(T^*)^2P^*_n(T^*)
\\[2ex]
\quad \displaystyle =\sum_{T^*\in \TT^*_n}\Big(\sum_{Q\in\mathcal{P}_4(L(T^*))} \sum_{i=1}^4 q_i\delta(Q;Q_i^*;T^*)\Big)^2P^*_n(T^*)\\ 
\quad \displaystyle =\sum_{T^*\in \TT^*_n}\Big(\sum_{Q\in\mathcal{P}_4(L(T^*))} \sum_{i=1}^4 q_i^2\delta(Q;Q_i^*;T^*)^2\Big)P^*_n(T^*)\\ 
\displaystyle\qquad
+\hspace*{-1ex}\sum_{T^*\in \TT^*_n} \Bigg[\sum_{(Q,Q')\in\mathcal{P}_4(L(T^*))^2\atop Q\neq Q'}
\sum_{(i,j)\in [4]^2} \hspace*{-1ex} q_iq_j \delta(Q;Q_i^*;T^*)\delta(Q';Q_j^*;T^*)\Bigg]P^*_n(T^*)
\end{array}
$$
Now, since $\delta(Q;Q_i^*;T^*)^2=\delta(Q;Q_i^*;T^*)$,
$$
\begin{array}{rl}
S_1& \displaystyle :=
\sum_{T^*\in \TT^*_n}\Big(\sum_{Q\in\mathcal{P}_4(L(T^*))} 
\sum_{i=1}^4 q_i^2\delta(Q;Q_i^*;T^*)^2\Big)P^*_n(T^*)\\ 
& \displaystyle =
\sum_{T^*\in \TT^*_n}\Big(\sum_{Q\in\mathcal{P}_4(L(T^*))} 
\sum_{i=1}^4 q_i^2\delta(Q;Q_i^*;T^*)\Big)P^*_n(T^*)\\ 
\\ 
& \displaystyle =\sum_{i=1}^4 \Big(q_i^2\sum_{T^*\in \TT^*_n}\big|\{Q\in\mathcal{P}_4(L(T^*)): T^*(Q)=Q_i^*\}\big |\cdot P^*_n(T^*) \Big)
\\ 
& \displaystyle =\binom{n}{4}\sum_{i=1}^4 \Big(q_i^2\sum_{T^*\in \TT^*_n}\frac{\big|\{Q\in\mathcal{P}_4(L(T^*)): T^*(Q)=Q_i^*\}\big |}{\binom{n}{4}}P^*_n(T^*)\Big) \\
\\ 
& \displaystyle =\binom{n}{4}\sum_{i=1}^4 q_i^2 P_4^*(Q_i^*)
\end{array} 
$$
by the sampling consistency of $P_n^*$.

As far as the second addend in the previous expression for $E_{P^*}({\mathit{rQI}^*_n}^2)$ goes, we have
$$
\begin{array}{l}
S_2  :=\displaystyle\sum_{T^*\in \TT^*_n}\sum_{(Q,Q')\in\mathcal{P}_4(L(T^*))^2\atop Q\neq Q'}
\Big(\sum_{(i,j)\in [4]^2} q_iq_j \delta(Q;Q_i^*;T^*)\delta(Q';Q_j^*;T^*)\Big)P^*_n(T^*)\\
\ =\displaystyle \sum_{(i,j)\in [4]^2}q_iq_j\Bigg[\sum_{T^*\in \TT^*_n}\Big(\sum_{k=0}^3 \sum_{(Q,Q')\in\mathcal{P}_4(L(T^*))^2\atop |Q\cap Q'|=k}
\hspace*{-0.7cm} \delta(Q;Q_i^*;T^*)\delta(Q';Q_j^*;T^*)\Big) P^*_n(T^*)\Bigg]\\
\ =\displaystyle \sum_{(i,j)\in [4]^2}q_iq_j\Bigg[\sum_{k=0}^3\sum_{T^*\in \TT^*_n} \big|\{(Q,Q')\in\mathcal{P}_4(L(T^*))^2: |Q\cap Q'|=k, \\[-2ex]
\ \displaystyle \hspace*{4.5cm} T^*(Q)=Q_i^*, T^*(Q')=Q_j^*\}\big|\cdot P^*_n(T^*)\Bigg]
\end{array}
$$
Now notice that, for every $k=0,\ldots,3$,
$$
\begin{array}{l}
\displaystyle \sum_{T^*\in \TT^*_n} \big|\{(Q,Q')\in\mathcal{P}_4(L(T^*))^2: |Q\cap Q'|=k, T^*(Q)=Q_i^*, T^*(Q')=Q_j^*\}\big| P^*_n(T^*)\\
\displaystyle \qquad=
\sum_{T^*\in \TT^*_n}\Big(\sum_{T^*_{8-k}\in \TT_{8-k}^*} \big|\{X \in\mathcal{P}_{8-k}(L(T^*)): T^*(X)=T^*_{8-k}\}\big|\\
\displaystyle \qquad\qquad
\cdot\big|\{(Q,Q')\in\mathcal{P}_4(L(T^*_{8-k}))^2: |Q\cap Q'|=k, T^*_{8-k}(Q)=Q_i^*, T^*_{8-k}(Q')=Q_j^*\}\big|\Big) P^*_n(T^*)\\
\displaystyle \qquad=
\sum_{T^*_{8-k}\in \TT_{8-k}^*}\big|\{(Q,Q')\in\mathcal{P}_4(L(T^*_{8-k}))^2: |Q\cap Q'|=k, T^*_{8-k}(Q)=Q_i^*, T^*_{8-k}(Q')=Q_j^*\}\big|\\
\displaystyle \qquad\qquad
\cdot\binom{n}{8-k}\sum_{T^*\in \TT^*_n}\frac{\Big( \big|\{X \in\mathcal{P}_{8-k}(L(T^*)): T^*(X)=T^*_{8-k}\}\big|}{\binom{n}{8-k}} P^*_n(T^*)\\
\displaystyle \qquad=
\binom{n}{8-k}\sum_{T^*_{8-k}\in \TT_{8-k}^*}\big|\{(Q,Q')\in\mathcal{P}_4(L(T^*_{8-k}))^2: |Q\cap Q'|=k, \\
\displaystyle \hspace*{4cm}  T^*_{8-k}(Q)=Q_i^*, T^*_{8-k}(Q')=Q_j^*\}\big|P^*_{8-k}(T_{8-k}^*)\\
\displaystyle \qquad=
\binom{n}{8-k}\sum_{T^*_{8-k}\in \TT_{8-k}^*}\big|\{(Q,Q')\in\mathcal{P}_4(L(T^*_{8-k}))^2: Q\cup Q'=L(T_{8-k}^*), \\
\displaystyle \hspace*{4cm}  T^*_{8-k}(Q)=Q_i^*, T^*_{8-k}(Q')=Q_j^*\}\big|P^*_{8-k}(T_{8-k}^*)\\
\displaystyle \qquad=
\binom{n}{8-k}\sum_{T^*_{8-k}\in \TT_{8-k}^*}\Theta_{i,j}(T^*_{8-k})P^*_{8-k}(T_{8-k}^*)
\end{array}
$$
again by the sampling consistency of $P_n^*$.
Therefore,
$$
\begin{array}{rl}
S_2  & =\displaystyle \sum_{(i,j)\in [4]^2}q_iq_j\Big(\sum_{k=0}^3\binom{n}{8-k} \sum_{T^*_{8-k}\in \TT_{8-k}^*}\Theta_{i,j}(T^*_{8-k})P^*_{8-k}(T_{8-k}^*)\Big)\\
& =\displaystyle \sum_{(i,j)\in [4]^2}q_iq_j\Big(\sum_{k=5}^8\binom{n}{k} \sum_{T^*\in \TT_{k}^*}\Theta_{i,j}(T^*)P^*_{k}(T^*)\Big)
\end{array}
$$
 
 The formula in the statement is then obtained by writing $Var_{P^*}(IQ^*_n)$ as $S_1+S_2-E_{P^*}(IQ^*_n)^2$ and using the expression for  $E_{P^*}(IQ^*_n)=E_{P}(IQ_n)$ given in Proposition \ref{thm:EQI}.
  \qed
\end{proof}

Again, if $P_n$ is a probabilistic model of bifurcating phylogenetic trees, so that 
$P^*_4(Q_1^*)=P^*_4(Q_2^*)=P^*_4(Q_4^*)=0$, then, taking $q_3=1$, this proposition implies that
$$
\begin{array}{l}
\displaystyle Var_{P}(\mathit{rQIB}_n) =
\binom{n}{4}P_4^*(Q_3^*)- \binom{n}{4}^2P^*_4(Q_3^*)^2\\
\displaystyle \qquad\quad+\sum_{k=5}^8\binom{n}{k}\Big( \sum_{T^*\in \BT_{k}^*}\Theta_{3,3}(T^*) P^*_{k}(T^*)\Big)
\end{array}
$$
In this bifurcating case, the figures $\Theta_{3,3}(T^*)$ appearing in this expression can be easily computed by hand: they are provided in  Table \ref{table:1varbin} in the Appendix A.2.  We obtain then the following result.

\begin{corollary}\label{cor:varQIB}
If $P_n$ is a probabilistic model of bifurcating phylogenetic trees such that  $P_n^*$ is sampling consistent, then, with the notations for trees given in Table \ref{table:1varbin} in the Appendix A.2,
$$
\begin{array}{l}
\displaystyle Var_{P}(\mathit{rQIB}_n) =\binom{n}{4}P_4^*(Q_3^*)- \binom{n}{4}^2P^*_4(Q_3^*)^2\\
\displaystyle \qquad\quad+
6\binom{n}{5}P_5^*(B_{5,3}^*)+\binom{n}{6}\big(18P_{6}^*(B_{6,4}^*)+6P_{6}^*(B_{6,5}^*)+36P_{6}^*(B_{6,6}^*)\big)\\
\displaystyle \qquad\quad+
\binom{n}{7}\big(8P_7^*(B_{7,8}^*)+24P_7^*(B_{7,9}^*)+36P_7^*(B_{7,10}^*)+36P_7^*(B_{7,11}^*)\big)\\
\displaystyle \qquad\quad+
\binom{n}{8}\big(2P_8^*(B_{8,13}^*)+6P_8^*(B_{8,14}^*)+12P_8^*(B_{8,15}^*)+14P_8^*(B_{8,16}^*)\\
\displaystyle \qquad\quad\qquad
+18P_8^*(B_{8,17}^*)+36P_8^*(B_{8,21}^*)+36P_8^*(B_{8,22}^*)+38P_8^*(B_{8,23}^*)\big)
\end{array}
$$
\end{corollary}

Proposition \ref{prop:varqi} and Corollary \ref{cor:varQIB} reduce the computation of $Var_{P}(\mathit{rQI}_n)$  or $Var_{P}(\mathit{rQIB}_n)$ to the explicit knowledge of $P^*_{l}$ for $l=4,5,6,7,8$.  In particular, they allow to obtain explicit formulas for the variance of $\mathit{rQIB}_n$ under the $\alpha$ and the $\beta$-models, and  for the variance of $\mathit{rQI}_n$ under the $\alpha$-$\gamma$-model. 

As far as the bifurcating case goes, on the one hand, the probabilities under the $\alpha$-model of the trees appearing explicitly in the formula for the variance of $\mathit{rQIB}_n$ in Corollary \ref{cor:varQIB} are those  given in Table  \ref{table:2varbin} in the Appendix A.2 (they are explicitly computed in the Supplementary Material of \citep{MMR}).
Plugging them in the formula given in Corollary \ref{cor:varQIB} above, we obtain the following result.

\begin{corollary}\label{cor:varQIBalfa}
Under the $\alpha$-model,
$$
\begin{array}{l}
\displaystyle Var_{P^F_\alpha}(\mathit{rQIB}_n) = \binom{n}{4}\frac{1-\alpha}{3-\alpha}- \binom{n}{4}^2\frac{(1-\alpha)^2}{(3-\alpha)^2}+12\binom{n}{5}\frac{1-\alpha}{4-\alpha}\\
\displaystyle \qquad\quad+
\binom{n}{6}\frac{6(1-\alpha)(112-89\alpha+15\alpha^2)}{(5-\alpha)(4-\alpha)(3-\alpha)}\\
\displaystyle \qquad\quad+
\binom{n}{7}\frac{20(1-\alpha)(74 - 63 \alpha +7 \alpha^2 )}{(6-\alpha)(5-\alpha)(3-\alpha)}\\
\displaystyle \qquad\quad +
\binom{n}{8}\frac{10(1-\alpha)(506-539\alpha+112\alpha^2-7\alpha^3)}{(7-\alpha)(6-\alpha)(5-\alpha)(3-\alpha)}
\end{array}
$$
\end{corollary}

The leading term in $n$ of $Var_{P^F_\alpha}(\mathit{rQIB}_n)$ is then
$$
\frac{(1-\alpha)(2\alpha+1)}{84(7-\alpha)(6-\alpha)(5-\alpha)(3-\alpha)^2}\cdot n^8.
$$

On the other hand, the probabilities under the $\beta$-model of the same trees are  given in Table  \ref{table:3varbin} in the Appendix A.2, yielding the following result.

\begin{corollary}\label{cor:varQIBeta}
Under the $\beta$-model,
$$
\begin{array}{l}
\displaystyle Var_{P^A_\beta}(\mathit{rQIB}_n) = \binom{n}{4}\dfrac{3(\beta+2)}{7\beta+18}- \binom{n}{4}^2\dfrac{9(\beta+2)^2}{(7\beta+18)^2}
+12\binom{n}{5}\dfrac{\beta+2}{3\beta+8}
\\
\displaystyle \qquad\quad+
90\binom{n}{6}\frac{(\beta + 2) (41 \beta^2 + 238 \beta + 336)}{(31 \beta^2 + 194 \beta + 300)(7 \beta + 18)}
\\
\displaystyle \qquad\quad+
60\binom{n}{7}\frac{(\beta+2)(9 \beta^2 + 53 \beta+ 74)}{(\beta + 3)(3\beta+10) (7 \beta + 18)}
\\
\displaystyle \qquad\quad +
630\binom{n}{8}\frac{(\beta+2)(127 \beta^4 + 1637 \beta^3 + 7788 \beta^2 + 16084 \beta + 12144)}{(127\beta^3+1383\beta^2+4958\beta+5880)(7\beta + 18)^2}
\end{array}
$$
\end{corollary}

So, the leading term in $n$ of $Var_{P^A_\beta}(\mathit{rQIB}_n)$ is
$$
\frac{(\beta + 2) (2 \beta^2 + 9 \beta + 12)}{2 (7\beta + 18)^2 (127 \beta^3 + 1383 \beta^2 + 4958 \beta + 5880)}\cdot n^8.
$$

When $\alpha=0$ or $\beta=0$, which correspond to the Yule model, both formulas for the variance of $\mathit{rQIB}_n$ reduce to
$$
Var_Y(\mathit{rQIB}_n)=\binom{n}{4}\frac{5 n^4 + 30 n^3 + 118 n^2 + 408 n + 630}{33075}.
$$
In the Appendix A.1 we give an independent derivation of this formula, which provides extra evidence of the correctness of all these computations.

As far as the uniform model goes, when $\alpha=1/2$ or  $\beta=0$, both formulas yield
$$
Var_U(\mathit{rQIB}_n)=\binom{n}{4}\frac{4(2 n - 1) (2 n + 1) (2 n + 3) (2 n + 5)}{225225}.
$$

Finally, as far as the $\alpha$-$\gamma$-model goes, we have written a set of Python scripts that compute all $\Theta_{i,j}(T^*)$, $i,j=1,2,3,4$, as well as $P_{\alpha,\gamma,k}^*(T^*)$ for every $T^*\in \TT_k^*$, $k=5,6,7,8$, and combine all these data into an explicit formula for  $Var_{P_{\alpha,\gamma}}(\mathit{rQI}_n)$. The Python scripts and the resulting formula (in text format and as a Python script that can be applied to any values of $n$, $\alpha$, and $\gamma$) can be found in the GitHub page \url{https://github.com/biocom-uib/Quartet_Index} companion to this paper. In particular, the plain text formula (which is too long and uninformative  to be reproduced here) is given in the document  \texttt{variance\_{}table.txt} therein. It can be easily checked using a symbolic computation program that when $\alpha=\gamma$ it agrees with  the variance under the $\alpha$-model given in Corollary \ref{cor:varQIBalfa}.

\section{Conclusions}

In this paper we have introduced a new balance index for phylogenetic trees, the rooted quartet index $\mathit{rQI}$. This index makes sense for multifurcating  trees, it can be computed in time linear in the number of leaves, and it has a larger range of values than any other shape index defined so far. We have computed its maximum and minimum values for bifurcating and arbitrary trees, and we have shown how to compute its expected value and variance under any probabilistic model of phylogenetic trees that is sampling consistent and invariant under relabelings. This includes the popular uniform, Yule, $\alpha$, $\beta$ and $\alpha$-$\gamma$-models. This paper is accompanied by the GitHub page~\url{https://github.com/biocom-uib/Quartet_Index} where the interested reader can find a set of Python scripts that perform several computations related to this index.

We want to call the reader's attention on a further property of the rooted quartet index: it can be used in a sensible way to measure the balance of \emph{taxonomic trees}, defined as those rooted trees of fixed depth (but with possibly out-degree 1 internal nodes) with their leaves bijectively labeled in a set of taxa.
The usual taxonomies with fixed ranks are the paradigm of such taxonomic trees.
It turns out that the classical balance indices cannot be used in a sound way to quantify the balance of such trees. For instance, Colless' index cannot be applied to multifurcating trees, and Sackin's index, being the sum of the depths of the leaves in the tree, is constant on all taxonomic trees of fixed depth and number of leaves. As far as the total cophenetic index goes, it is straightforward to check from its very definition that the taxonomic trees with maximum and minimum total cophenetic values among all taxonomic trees of a given depth and a given number of leaves are those depicted in Fig~\ref{fig:cophtax}. In our opinion, these two trees should be considered as equally balanced. We believe that  $\mathit{rQI}$  can be used to capture the symmetry of a taxonomic tree in a natural way, and we hope to report on it elsewhere. 

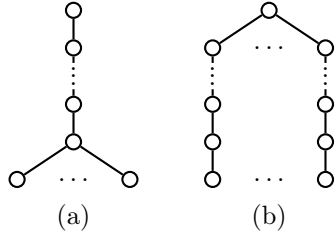
\begin{figure}[htb]
\begin{center}
\begin{tikzpicture}[thick,>=stealth,scale=0.25]
\draw(0,0) node [trep] (1) {}; 
\draw(3,0) node  {$\ldots$};  
\draw(6,0) node [trep] (n) {}; 
\draw(3,2) node[trep] (a) {};
\draw(3,4) node[trep] (b) {};
\draw(3,5.1) node  {.};
\draw(3,5.5) node  {.};
\draw(3,5.9) node  {.};
\draw(3,7) node[trep] (c) {};
\draw(3,9) node[trep] (r) {};
\draw (a)--(1);
\draw (a)--(n);
\draw (b)--(a);
\draw (b)-- (3,4.8);
\draw (c)-- (3,6.2);
\draw (r)-- (c);
\draw(3,-2) node {(a)};
\end{tikzpicture}
\qquad
\begin{tikzpicture}[thick,>=stealth,scale=0.25]
\draw(0,0) node [trep] (1) {}; 
\draw(3,0) node  {$\ldots$};  
\draw(6,0) node [trep] (n) {}; 
\draw(0,2) node[trep] (a1) {};
\draw(0,4) node[trep] (b1) {};
\draw(0,5.1) node  {.};
\draw(0,5.5) node  {.};
\draw(0,5.9) node  {.};
\draw(0,7) node[trep] (c1) {};
\draw(6,2) node[trep] (a2) {};
\draw(6,4) node[trep] (b2) {};
\draw(6,5.1) node  {.};
\draw(6,5.5) node  {.};
\draw(6,5.9) node  {.};
\draw(6,7) node[trep] (c2) {};
\draw(3,7) node  {$\ldots$};  
\draw(3,9) node[trep] (r) {};
\draw (a1)--(1);
\draw (a2)--(n);
\draw (b1)--(a1);
\draw (b2)--(a2);
\draw (b1)-- (0,4.8);
\draw (c1)-- (0,6.2);
\draw (b2)-- (6,4.8);
\draw (c2)-- (6,6.2);
\draw (r)-- (c1);
\draw (r)-- (c2);
\draw(3,-2) node {(b)};
\end{tikzpicture}
\end{center}
\caption{\label{fig:cophtax} 
The shapes of the taxonomic trees with maximum (a) and minimum (b) total cophenetic values among all taxonomic trees of given depth and number of leaves.}
\end{figure}

{In a future paper we also plan to study some further properties of $\mathit{rQI}$, like for instance its correlation with other balance indices under different probabilistic models.  To illustrate the relation between $\mathit{rQI}$ and other balance indices, in Fig. \ref{fig:15} we provide  scatterplots of the values of $\mathit{rQI}$  (taking $q_i=i$) \textsl{versus} the Sackin index $S$, the Colless index $C$, the total cophenetic index $\Phi$ and the number of cherries on $\BT_{20}$ (which contains more than $(2\cdot 20-3)!!\geq 8.2 \times 10^{21}$ members) and \textsl{versus} the Sackin index $S$ and the total cophenetic index $\Phi$ on $\TT_{15}$ (which contains more than $6.3\times 10^{15}$ members). The values of the Spearman correlations between these pairs of indices on these classes of trees are given in Table \ref{tab:Spearman}. We want to point out the small correlation between $\mathit{rQI}$ and the number of cherries: although at first sight it could seem that  counting the number of fully symmetric rooted quartets in a tree is equivalent to counting pairs of cherries, it is not the case, as the cherries in a rooted quartet may correspond to distant leaves in the tree. For instance, the bifurcating phylogenetic trees in Fig. \ref{fig:nocherries} have both 2 cherries, but their $\mathit{rQI}$ value is quite different. Notice also that the correlations between $\mathit{rQI}$ and $S$, $C$ and $\Phi$ are negative, because $\mathit{rQI}$ grows while $S$, $C$ and $\Phi$ decrease with the balance of the trees.}

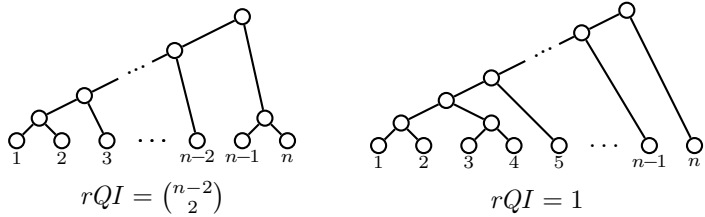
\begin{figure}[htb]
\begin{center}
\begin{tikzpicture}[thick,>=stealth,scale=0.3]
\draw(0,0) node [trep] (1) {};  
\draw(2,0) node [trep] (2) {};   
\draw(4,0) node [trep] (3) {};  
\draw(6,0) node  {$\ldots$};  
\draw(8,0) node [trep] (n2) { };   
\draw(10,0) node [trep] (n1) { };   
\draw(12,0) node [trep] (n) { };  
\draw(0,-0.75) node  {\scriptsize $1$};  
\draw(2,-0.75) node  {\scriptsize $2$};   
\draw(4,-0.75) node  {\scriptsize $3$};  
\draw(6,0) node  {$\ldots$};  
\draw(8,-0.75) node  {\scriptsize $n\!\!-\!\!2$};   
\draw(10,-0.75) node  {\scriptsize $n\!\!-\!\!1$};   
\draw(12,-0.75) node  {\scriptsize $n$};  
\draw(1,1) node[trep] (a) {};
\draw(3,2) node[trep] (b) {};
\draw(7,4) node[trep] (c) {};
\draw(10,5.5) node[trep] (r) {};
\draw(11,1) node[trep] (d) {};
\draw(5,3) node  {.};
\draw(5.3,3.15) node  {.};
\draw(5.6,3.3) node  {.};
\draw (a)--(1);
\draw (a)--(2);
\draw (b)--(3);
\draw (b)--(a);
\draw (b)-- (4.5,2.75);
\draw (c)-- (6,3.5);
\draw (c)-- (n2);
\draw (d)--(n1);
\draw (d)--(n);
\draw (r)--(d);
\draw (r)--(c);
\draw(6,-2.5) node {$\mathit{rQI}=\binom{n-2}{2}$};
\end{tikzpicture}
\qquad
\begin{tikzpicture}[thick,>=stealth,scale=0.3]
\draw(0,0) node [trep] (1) {};  
\draw(2,0) node [trep] (2) {};   
\draw(4,0) node [trep] (3) {};  
\draw(6,0) node [trep] (4) {};  
\draw(8,0) node [trep] (5) {};  
\draw(10,0) node  {$\ldots$};  
\draw(12,0) node [trep] (n1) { };   
\draw(14,0) node [trep] (n) { };  
\draw(0,-0.75) node  {\scriptsize $1$};  
\draw(2,-0.75) node  {\scriptsize $2$};   
\draw(4,-0.75) node  {\scriptsize $3$};  
\draw(6,-0.75) node  {\scriptsize $4$};  
\draw(8,-0.75) node  {\scriptsize $5$};  
\draw(12,-0.75) node  {\scriptsize $n\!\!-\!\!1$};   
\draw(14,-0.75) node  {\scriptsize $n$};  
\draw(1,1) node[trep] (a1) {};
\draw(5,1) node[trep] (a2) {};
\draw (a1)--(1);
\draw (a1)--(2);
\draw (a2)--(3);
\draw (a2)--(4);
\draw(3,2) node[trep] (b) {};
\draw (b)--(a1);
\draw (b)--(a2);
\draw(5,3) node[trep] (c) {};
\draw (c)--(b);
\draw (c)--(5);
\draw(7,4) node  {.};
\draw(7.3,4.15) node  {.};
\draw(7.6,4.3) node  {.};
\draw(9,5) node[trep] (d) {};
\draw(11,6) node[trep] (r) {};
\draw (d)--(n1);
\draw (r)--(d);
\draw (r)--(n);
\draw (c)-- (6.5,3.75);
\draw (d)-- (8,4.5);
\draw(7,-2.5) node {$\mathit{rQI}=1$};
\end{tikzpicture}
\end{center}
\caption{\label{fig:nocherries} 
Two trees with 2 cherries and very different $\mathit{rQI}$.}
\end{figure}

\begin{figure}[htb]
\begin{center}
\begin{tabular}{cc}
\includegraphics[width=0.4\linewidth]{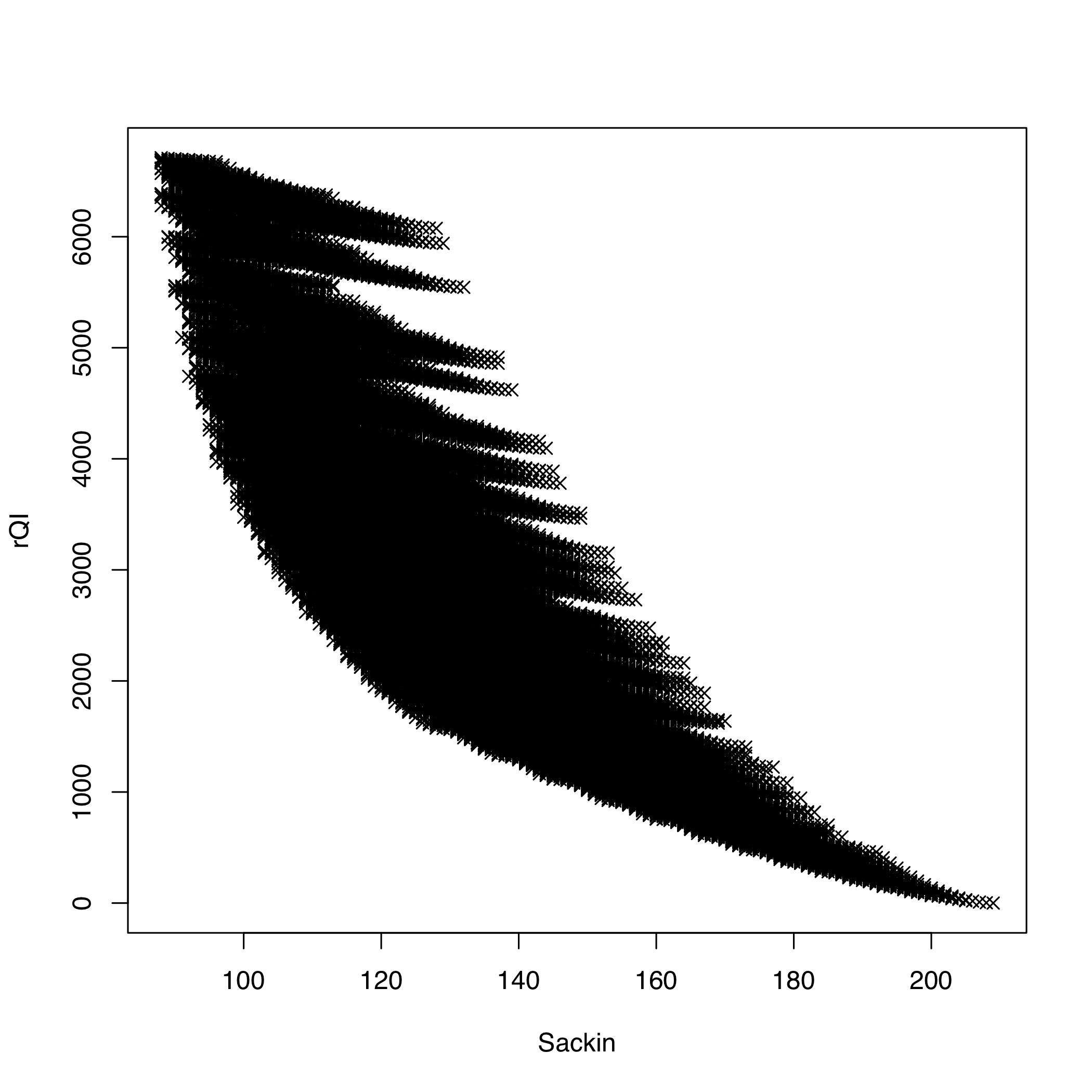} & 
\includegraphics[width=0.4\linewidth]{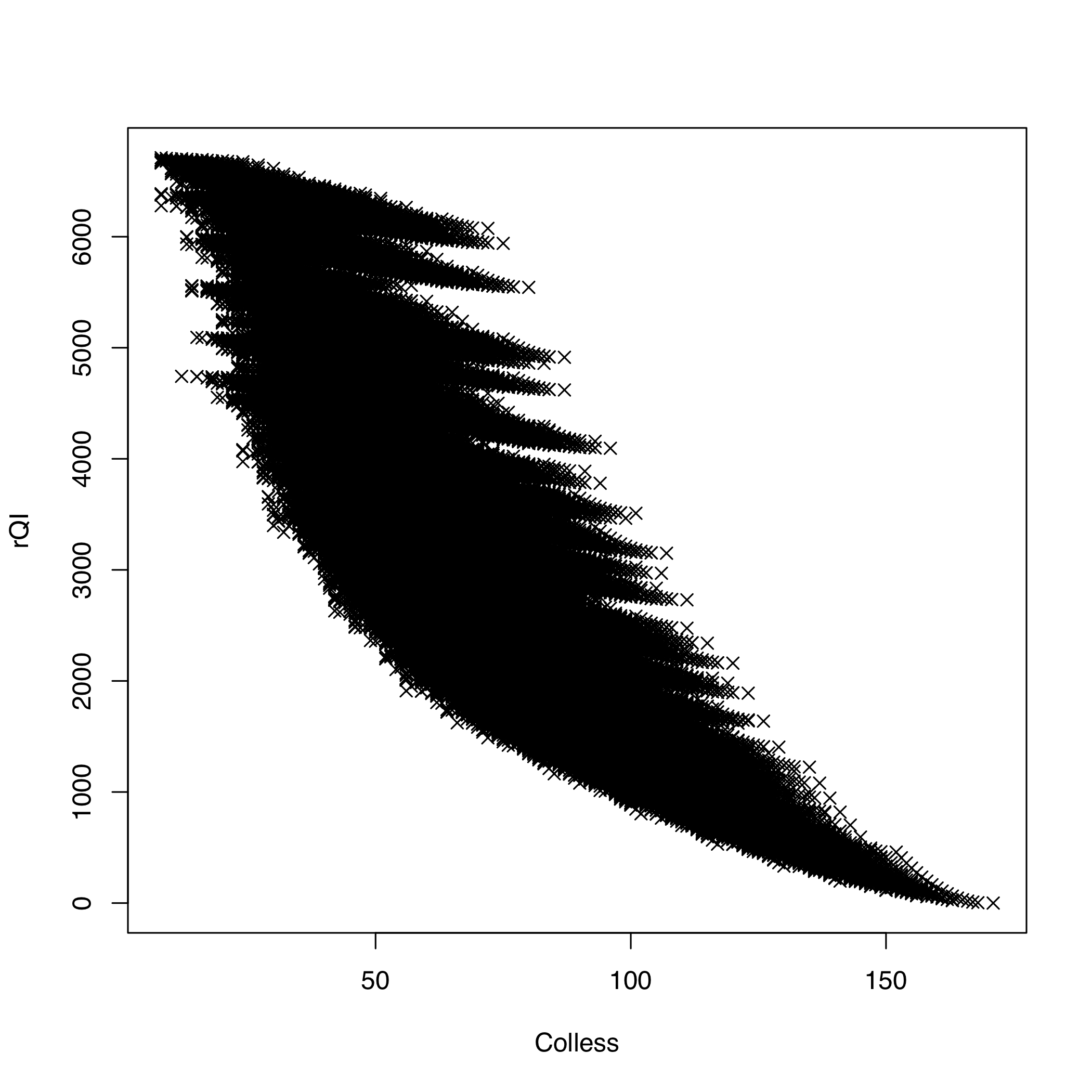}  
\\
(a) & (b)
\\
\includegraphics[width=0.4\linewidth]{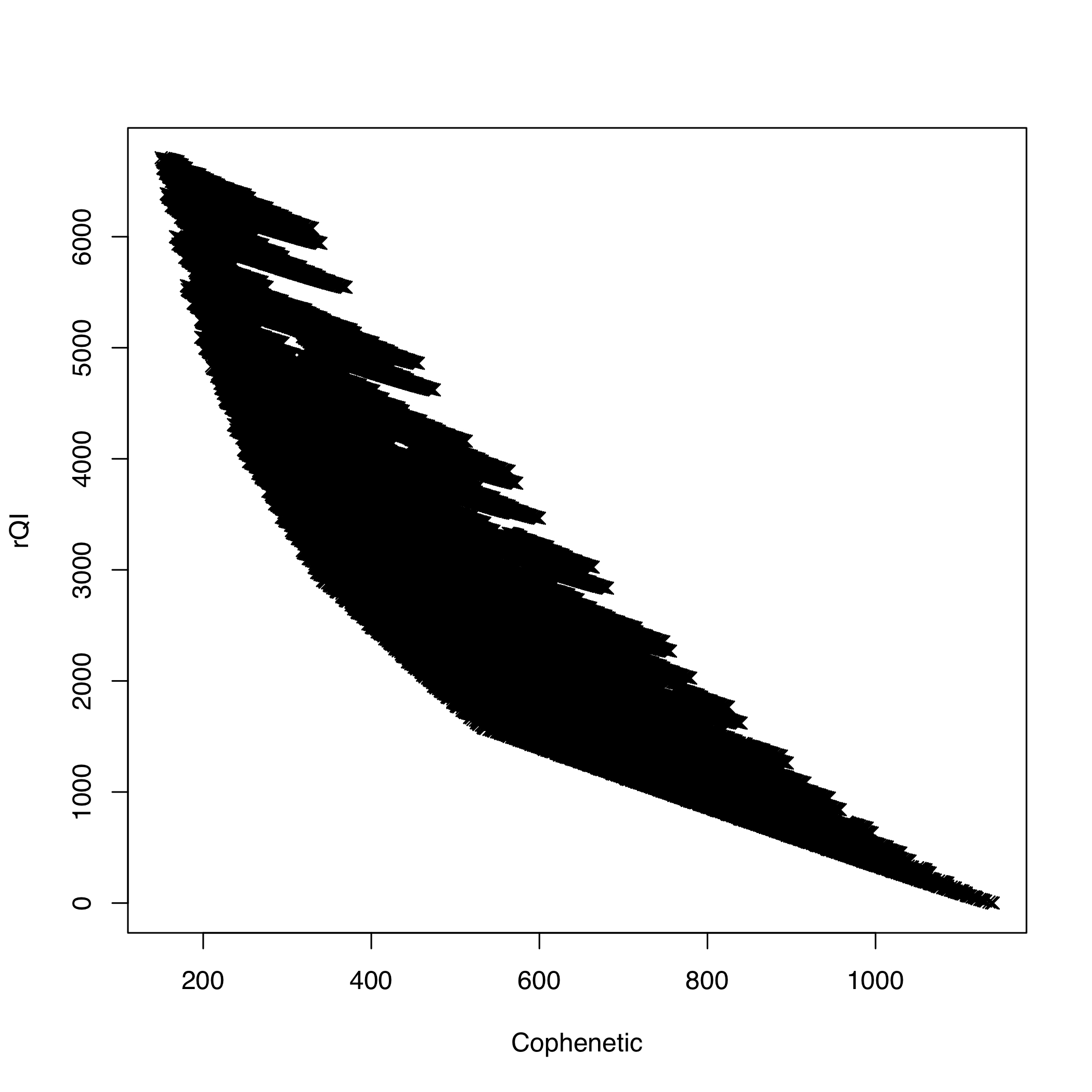} & 
\includegraphics[width=0.4\linewidth]{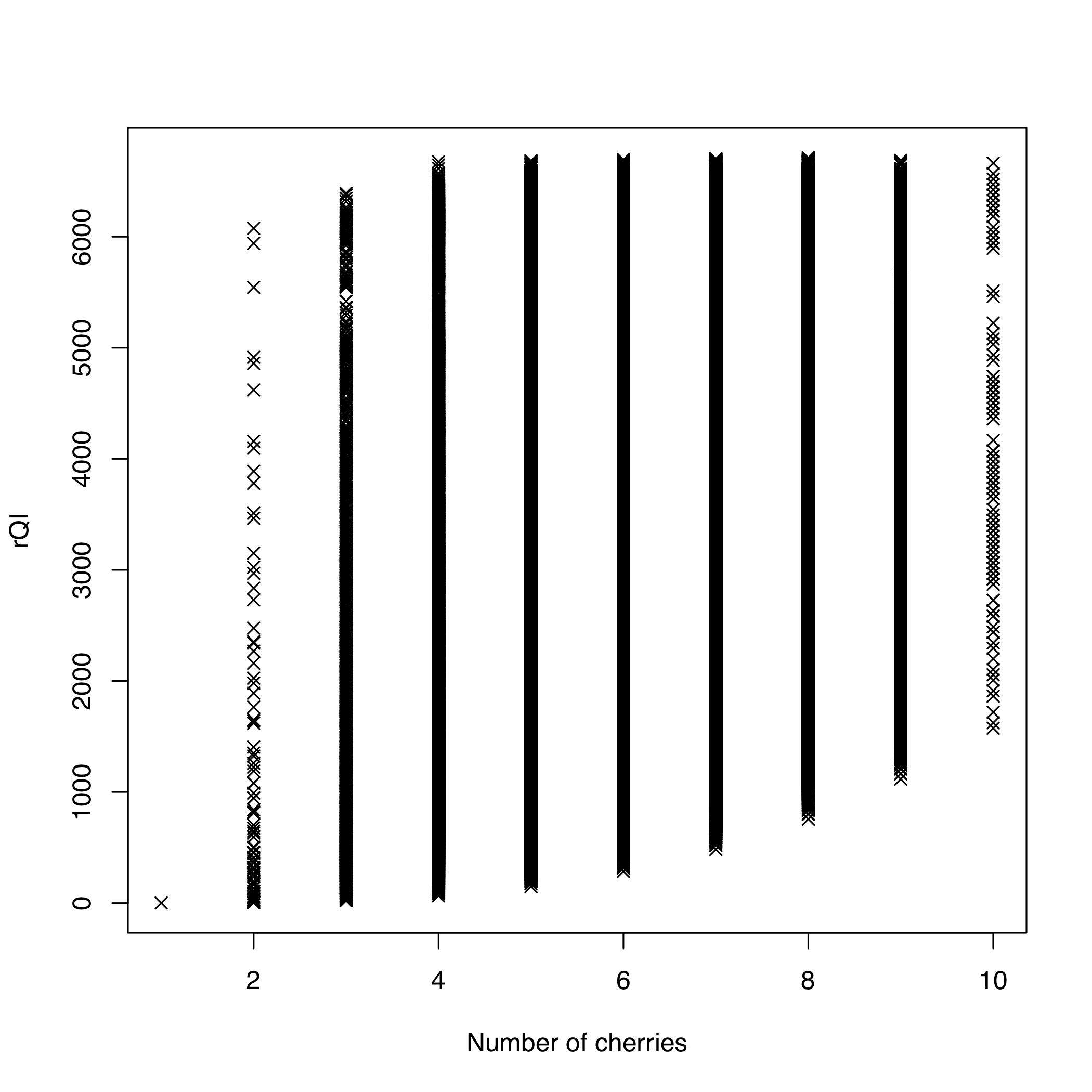}  
\\
(c) & (d)
\\
\includegraphics[width=0.4\linewidth]{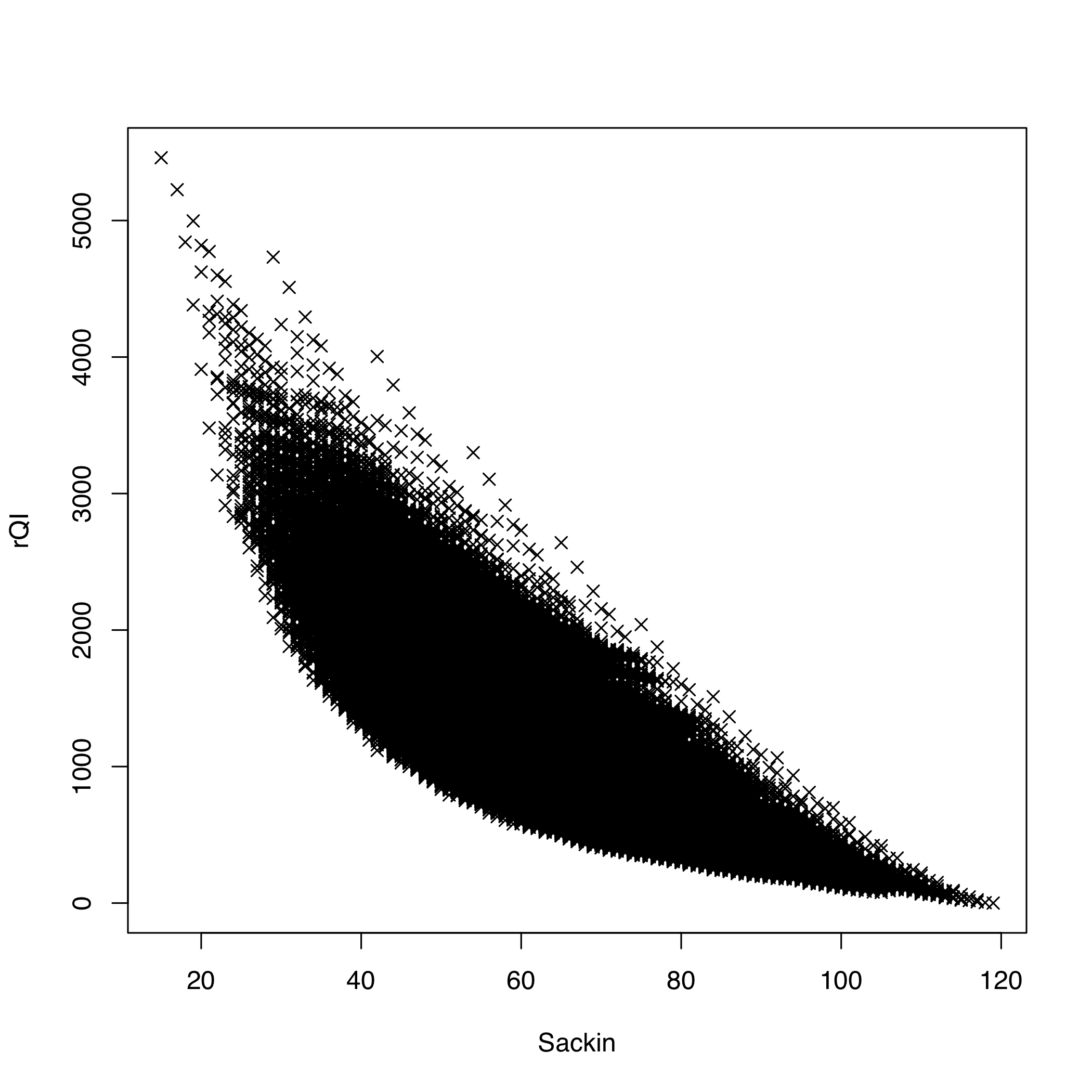} & 
\includegraphics[width=0.4\linewidth]{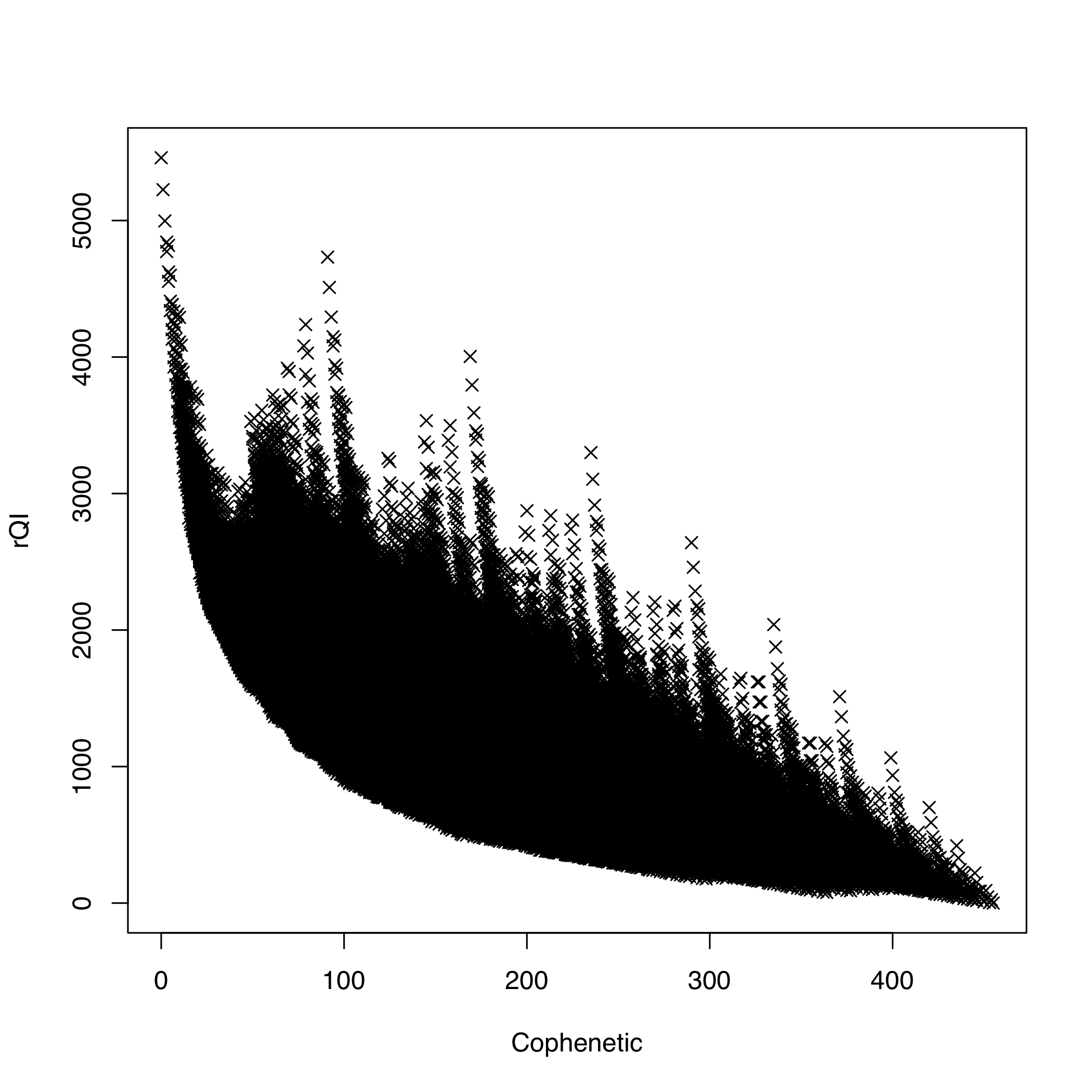}  
\\
(e) & (f)
\end{tabular}
\end{center}
\caption{\label{fig:15} 
Scatterplot of $\mathit{rQI}$ \textsl{versus}: (a) the Sackin index on $\BT_{20}$;
(b) the Colless index on $\BT_{20}$; (c) the total cophenetic index  on $\BT_{20}$; (d) 
the number of cherries on $\BT_{20}$; (e) the Sackin index on $\TT_{15}$;
(f) the total cophenetic index  on $\TT_{15}$.}
\end{figure}

\begin{table}[h]
\begin{center}\small
\begin{tabular}{l|r}
Correlation & Value \\ \hline
$\mathit{rQI}$ \textsl{vs} $S$ on $\BT_{20}$ &  $-0.889$
\\
$\mathit{rQI}$ \textsl{vs} $C$ on $\BT_{20}$ &  $-0.893$
\\
$\mathit{rQI}$ \textsl{vs} $\Phi$ on $\BT_{20}$ &  $ -0.935$
\\
$\mathit{rQI}$ \textsl{vs} $\mathit{Ch}$ on $\BT_{20}$ & $0.165$
\\
$\mathit{rQI}$ \textsl{vs} $S$ on $\TT_{15}$ & $-0.787$
\\
$\mathit{rQI}$ \textsl{vs} $\Phi$ on $\TT_{15}$ &  $-0.827$
\\
\end{tabular}
\end{center}
\caption{\label{tab:Spearman} Spearman's correlations corresponding to the scatterplots in Fig. \ref{fig:15}.}
\end{table}

\noindent\textbf{Acknowledgements.}
A preliminary version of this paper was presented at the \textsl{Workshop on Algebraic and combinatorial phylogenetics} held in Barcelona (June 26--30, 2017). We thank Mike Steel, Gabriel Riera, Seth Sullivant, the anonymous reviewers and the associate editor for their helpful suggestions on several aspects of this paper. This research was partially supported by the Spanish Ministry of Economy and Competitiveness and the European Regional Development Fund through project DPI2015-67082-P (MINECO /FEDER).

\section*{Appendices}

\subsection*{A.1: An alternative derivation of the variance of $\mathit{rQIB}_n$ under the Yule model}

In this section we give an alternative proof of the following result.

\begin{proposition}
Under the Yule model,
$$
Var_{Y}(\mathit{rQIB}_n)=\binom{n}{4}\frac{5 n^4+ 30 n^3+ 118 n^2 + 408 n+ 630}{33075}.
$$
\end{proposition}

\begin{proof}
By Lemma \ref{lem:rec}, $\mathit{rQIB}$ on $\BT_n$ is a bifurcating recursive tree shape statistic satisfying the recurrence
$$
\mathit{rQIB}(T\star T')=\mathit{rQIB}(T)+\mathit{rQIB}(T')+f_{\mathit{rQIB}}(|L(T)|,|L(T')|)
$$
with $f_{\mathit{rQIB}}(a,b)=\binom{a}{2}\binom{b}{2}$. 
Then, it satisfies the hypothesis in Cor. 1 of \citep{CMR} with
$$
\begin{array}{rl}
\displaystyle \varepsilon(a,b-1) &\displaystyle =f_{\mathit{rQIB}}(a,b)-f_{\mathit{rQIB}}(a,b-1)\\ & \displaystyle =\binom{a}{2}\binom{b}{2}-\binom{a}{2}\binom{b-1}{2}=(b-1)\binom{a}{2}\\[2ex]
\displaystyle R(n-1) & \displaystyle=E_{Y}(\mathit{rQIB}_n)-E_{Y}(\mathit{rQIB}_{n-1})\\ & \displaystyle =\frac{1}{3}\binom{n}{4}-\frac{1}{3}\binom{n-1}{4}=\frac{1}{3}\binom{n-1}{3}
\end{array}
$$
Since $E_{Y}(\mathit{rQIB}_1)=0$ and $f_{\mathit{rQIB}}(n-1,1)=0$, applying the aforementioned result from \citep{CMR} we have that
$$
\begin{array}{l}
E_{Y}(\mathit{rQIB}_n^2)  \displaystyle = \frac{n}{n-1}  E_{Y}(\mathit{rQIB}_{n-1}^2)+\frac{4}{n-1} \sum_{k=1}^{n-2} \varepsilon(k,n-k-1) E_{Y}(\mathit{rQIB}_k) \\
\qquad\qquad \displaystyle +\frac{2}{n-1} \sum_{k=1}^{n-2}  R(n-k-1)E_{Y}(\mathit{rQIB}_k) \\
\qquad\qquad \displaystyle
+\frac{1}{n-1}\sum_{k=1}^{n-2} (f_{\mathit{rQIB}}(k,n-k)^2- f_{\mathit{rQIB}}(k,n-k-1)^2)\\[2ex]
 \qquad\displaystyle = \frac{n}{n-1}  E_{Y}(\mathit{rQIB}_{n-1}^2)+\frac{4}{3(n-1)} \sum_{k=1}^{n-2} (n-k-1)\binom{k}{2}\binom{k}{4} \\
\qquad\qquad \displaystyle +\frac{2}{9(n-1)} \sum_{k=1}^{n-2}  \binom{n-k-1}{3} \binom{k}{4} \\
\qquad\qquad \displaystyle
+\frac{1}{n-1}\sum_{k=1}^{n-2} \binom{k}{2}^2\Bigg(\binom{n-k}{2}^2-\binom{n-k-1}{2}^2\Bigg)\\[2ex]
\qquad \displaystyle = \frac{n}{n-1}  E_{Y}(\mathit{rQIB}_{n-1}^2) +\frac{n}{3}\binom{n-2}{4}\frac{15 n^2 - 35 n + 6}{420}\\[1ex]
\qquad\qquad \displaystyle +\frac{n}{9} \binom{n-2}{4}\frac{n^2 - 13 n + 42}{840}\\[1ex]
\qquad\qquad \displaystyle + n\binom{n-2}{2}\frac{3 n^4 - 18 n^3 + 41 n^2 - 42 n + 36}{1680}\\[2ex]
\qquad \displaystyle = \frac{n}{n-1}  E_{Y}(\mathit{rQIB}_{n-1}^2) \\[1ex]
\qquad\qquad \displaystyle + \frac{n(n-2)(n-3)(253 n^4-2014 n^3 +6119 n^2-7430 n+ 3504)}{181440}
\end{array}
$$
Dividing by $n$ both sides of this expression for $E_{Y}(\mathit{rQIB}^2_n)$ and setting $y_n=E_{Y}(\mathit{rQIB}^2_n)/n$,
we obtain the recurrence
$$
\begin{array}{rl}
y_n & =\displaystyle y_{n-1}+ \frac{(n-2)(n-3)(253 n^4-2014 n^3 +6119 n^2-7430 n+ 3504)}{181440}.
\end{array}
$$
Since $y_0=y_1=0$, its solution is
$$
\begin{array}{rl}
y_n & \displaystyle = \sum_{k=2}^n  \frac{(k-2)(k-3)(253 k^4-2014 k^3 +6119 k^2-7430 k+ 3504)}{181440}\\[1ex]
& \displaystyle = \frac{(n - 3) (n - 2) (n - 1) (1265 n^4 - 7110 n^3 + 14419 n^2 - 4086 n + 5040)}{6350400}
\end{array}
$$
from where we obtain
$$
\begin{array}{rl}
E_{Y}(\mathit{rQIB}_n^2)& =ny_n\\
&\displaystyle = \binom{n}{4}\frac{1265 n^4 - 7110 n^3 + 14419 n^2 - 4086 n + 5040}{264600}.
\end{array}
$$

Finally
$$
\begin{array}{l}
Var_{Y}(\mathit{rQIB}_n)   \displaystyle =E_{Y}(\mathit{rQIB}_n^2)-E_{Y}(\mathit{rQIB}_n)^2\\
\qquad \displaystyle=\binom{n}{4}\frac{1265 n^4 - 7110 n^3 + 14419 n^2 - 4086 n + 5040}{264600}-\frac{1}{9}\binom{n}{4}^2\\
\qquad \displaystyle=\binom{n}{4}\frac{5 n^4 + 30 n^3 + 118 n^2 + 408 n + 630}{33075},
\end{array}
$$
as we claimed. \qed 
\end{proof}

\subsection*{A.2: Some tables used in Section \ref{sec:QI}}

\begin{table}[h]
\begin{center}\small
\begin{tabular}{c|c|c}
Name& Shape & $\Theta_{3,3}(T^*)$ \\ \hline
$B_{5,1}^*$ & $(*,(*,(*,(*,*))))$ &  0
 \\ \hline
$B_{5,2}^*$ & $(*,((*,*),(*,*)))$ &   0
\\ \hline
$B_{5,3}^*$ & $((*,*),(*,(*,*)))$ & 6
\\ \hline \hline
$B_{6,1}^*$ & $(*,(*,(*,(*,(*,*)))))$ & 0 
 \\ \hline
$B_{6,2}^*$ & $(*,(*,((*,*),(*,*))))$ &   0
\\ \hline
$B_{6,3}^*$ & $(*,((*,*),(*,(*,*))))$ & 0
 \\ \hline
$B_{6,4}^*$ & $((*,*),((*,*),(*,*)))$ &  18
 \\ \hline
$B_{6,5}^*$ & $((*,*),(*,(*,(*,*))))$ &  6
 \\ \hline
$B_{6,6}^*$ & $((*,(*,*)),(*,(*,*)))$ &  36
\\ \hline \hline
$B_{7,1}^*$ & $(*,(*,(*,(*,(*,(*,*))))))$ &  0
\\ \hline
$B_{7,2}^*$ & $(*,(*,(*,((*,*),(*,*)))))$ &  0
\\ \hline
$B_{7,3}^*$ & $(*,(*,((*,*),(*,(*,*)))))$ &  0
\\ \hline
$B_{7,4}^*$ & $(*,((*,*),((*,*),(*,*))))$ &  0
\\ \hline
$B_{7,5}^*$ & $(*,((*,*),(*,(*,(*,*)))))$ &  0
\\ \hline
$B_{7,6}^*$ & $(*,((*,(*,*)),(*,(*,*))))$ &  0
\\ \hline
$B_{7,7}^*$ & $((*,*),(*,(*,(*,(*,*)))))$ &  0
\\ \hline
$B_{7,8}^*$ & $((*,*),(*,((*,*),(*,*))))$ &  8
\\ \hline
$B_{7,9}^*$ & $((*,*),((*,*),(*,(*,*))))$ &  24
\\ \hline
$B_{7,10}^*$ & $((*,(*,*)),(*,(*,(*,*))))$ &  36
\\ \hline
$B_{7,11}^*$ & $((*,(*,*)),((*,*),(*,*)))$ &  36
\\
\hline \hline  
$B_{8,1}^*$ & $(*,(*,(*,(*,(*,(*,(*,*)))))))$&  0
\\ \hline
$B_{8,2}^*$ & $(*,(*,(*,(*,((*,*),(*,*))))))$&  0
\\ \hline
$B_{8,3}^*$ & $(*,(*,(*,((*,*),(*,(*,*))))))$&  0
\\ \hline
$B_{8,4}^*$ & $(*,(*,((*,*),((*,*),(*,*)))))$&  0
\\ \hline
$B_{8,5}^*$ & $(*,(*,((*,*),(*,(*,(*,*))))))$&  0
\\ \hline
$B_{8,6}^*$ & $(*,(*,((*,(*,*)),(*,(*,*)))))$&  0
\\ \hline
$B_{8,7}^*$ & $(*,((*,*),(*,(*,(*,(*,*))))))$&  0
\\ \hline
$B_{8,8}^*$ & $(*,((*,*),(*,*),(*,(*,*)))))$&  0
\\ \hline
$B_{8,9}^*$ & $(*,((*,*),(*,((*,*),(*,*)))))$&  0
\\ \hline
$B_{8,10}^*$ & $(*,((*,(*,*)),(*,(*,(*,*)))))$&  0
\\ \hline
$B_{8,11}^*$ & $(*,((*,(*,*)),((*,*),(*,*))))$&  0
\\ \hline
$B_{8,12}^*$ & $((*,*),(*,(*,(*,(*,(*,*))))))$&  0
\\ \hline
$B_{8,13}^*$ & $((*,*),(*,(*,(*,*),(*,*))))$&  2
\\ \hline
$B_{8,14}^*$ & $((*,*),(*,((*,*),(*,(*,*)))))$&  6
\\ \hline
$B_{8,15}^*$ & $((*,*),((*,*),(*,(*,(*,*))))))$&  12
\\ \hline
$B_{8,16}^*$ & $((*,*),((*,*),((*,*),(*,*))))$&  14
\\ \hline
$B_{8,17}^*$ & $((*,*),((*,(*,*)),(*,(*,*))))$&  18
\\ \hline
$B_{8,18}^*$ & $((*,(*,*)),(*,(*,(*,(*,*)))))$&  0
\\ \hline
$B_{8,19}^*$ & $((*,(*,*)),(*,((*,*),(*,*)))))$&  0
\\ \hline
$B_{8,20}^*$ & $((*,(*,*)),((*,*),(*,(*,*))))$&  0
\\ \hline
$B_{8,21}^*$ & $((*,(*,(*,*))),(*,(*,(*,*))))$&  36
\\ \hline
$B_{8,22}^*$ & $((*,(*,(*,*))),((*,*),(*,*)))$&  36
\\ \hline
$B_{8,23}^*$ & $(((*,*),(*,*)),((*,*),(*,*)))$&  38
\\ \hline
\end{tabular}
\end{center}
\caption{\label{table:1varbin} Coefficients of  the probabilities of the trees in $\BT_k^*$, for $k=5,6,7,8$, in the formula for the variance of $\mathit{rQIB}_n$.}
\end{table}

\begin{table}[h]
\begin{center}\small
\begin{tabular}{c|c}
Tree& $P_{\alpha,n}^{A,*}$ \\ \hline\vphantom{$\Big($} 
$Q_3^*$ & $\frac{1-\alpha}{3-\alpha}$
\\ \hline\hline\vphantom{$\Big($} 
$B_{5,3}^*$ & $\frac{2(1-\alpha)}{4-\alpha}$
\\ \hline \hline\vphantom{$\Big($} 
$B_{6,4}^*$  &  $\frac{(1-\alpha)^2(8-\alpha)}{(5-\alpha)(4-\alpha)(3-\alpha)}$
 \\ \hline\vphantom{$\Big($} 
$B_{6,5}^*$  &  $\frac{2(1-\alpha)(8-\alpha)}{(5-\alpha)(4-\alpha)(3-\alpha)}$
 \\ \hline\vphantom{$\Big($} 
$B_{6,6}^*$  &  $\frac{2(1-\alpha)(2-\alpha)}{(5-\alpha)(4-\alpha)}$
\\ \hline \hline\vphantom{$\Big($} 
$B_{7,8}^*$  &  $\frac{(1-\alpha)^2(2+\alpha)(10+\alpha)}{(6-\alpha)(5-\alpha)(4-\alpha)(3-\alpha)}$
\\ \hline\vphantom{$\Big($} 
$B_{7,9}^*$  &  $\frac{2(1-\alpha)^2(10+\alpha)}{(6-\alpha)(5-\alpha)(4-\alpha)}$
\\ \hline\vphantom{$\Big($} 
$B_{7,10}^*$ &  $\frac{10(1-\alpha)(2-\alpha)}{(6-\alpha)(5-\alpha)(3-\alpha)}$
\\ \hline\vphantom{$\Big($} 
$B_{7,11}^*$  &  $\frac{5(1-\alpha)^2(2-\alpha)}{(6-\alpha)(5-\alpha)(3-\alpha)}$
\\
\hline \hline\vphantom{$\Big($}   
$B_{8,13}^*$ &  $\frac{8(1-\alpha)^2(1+\alpha)(2+\alpha)(3+\alpha)}{(7-\alpha)(6-\alpha)(5-\alpha)(4-\alpha)(3-\alpha)}$
\\ \hline\vphantom{$\Big($} 
$B_{8,14}^*$ &  $\frac{16(1-\alpha)^2(1+\alpha)(3+\alpha)}{(7-\alpha)(6-\alpha)(5-\alpha)(4-\alpha)}$
\\ \hline\vphantom{$\Big($} \vphantom{$\Big($} 
$B_{8,15}^*$ &  $\frac{8(1-\alpha)^2(3+\alpha)(8-\alpha)}{(7-\alpha)(6-\alpha)(5-\alpha)(4-\alpha)(3-\alpha)}$
\\ \hline\vphantom{$\Big($} 
$B_{8,16}^*$ &  $\frac{4(1-\alpha)^3(3+\alpha)(8-\alpha)}{(7-\alpha)(6-\alpha)(5-\alpha)(4-\alpha)(3-\alpha)}$
\\ \hline\vphantom{$\Big($} 
$B_{8,17}^*$ &  $\frac{8(1-\alpha)^2(2-\alpha)(3+\alpha)}{(7-\alpha)(6-\alpha)(5-\alpha)(4-\alpha)}$
\\ \hline\vphantom{$\Big($} 
$B_{8,21}^*$ &  $\frac{20(1-\alpha)(2-\alpha)}{(7-\alpha)(6-\alpha)(5-\alpha)(3-\alpha)}$
\\ \hline\vphantom{$\Big($} 
$B_{8,22}^*$ &  $\frac{20(1-\alpha)^2(2-\alpha)}{(7-\alpha)(6-\alpha)(5-\alpha)(3-\alpha)}$
\\ \hline\vphantom{$\Big($} 
$B_{8,23}^*$ &  $\frac{5(1-\alpha)^3(2-\alpha)}{(7-\alpha)(6-\alpha)(5-\alpha)(3-\alpha)}$
\\ \hline
\end{tabular}
\end{center}
\caption{\label{table:2varbin} Probabilities under the $\alpha$-model of the trees involved in the formula for the variance of $\mathit{rQIB}_n$}
\end{table}

\begin{table}[h]
\begin{center}\small
\begin{tabular}{c|c}
Tree & $P_{\beta,n}^{B,*}$ \\ \hline\vphantom{$\Big($} 
$Q_3^*$  &  $\frac{3(\beta+2)}{7\beta+18}$
\\ \hline\hline\vphantom{$\Big($} 
$B_{5,3}^*$  &   $\frac{2(\beta+2)}{3\beta+8}$
\\ \hline \hline\vphantom{$\Big($} 
$B_{6,4}^*$  & $\frac{45(\beta+2)^2(\beta+4)}{(31\beta^2+194\beta+300)(7\beta+18)}$  
\\ \hline\vphantom{$\Big($} 
$B_{6,5}^*$  &   $\frac{60(\beta+2)(\beta+3)(\beta+4)}{(31\beta^2+194\beta+300)(7\beta+18)}$
 \\ \hline\vphantom{$\Big($} 
$B_{6,6}^*$  &   $\frac{10(\beta+2)(\beta+3)}{31\beta^2+194\beta+300}$
\\ \hline \hline\vphantom{$\Big($} 
$B_{7,8}^*$ &  $\frac{3(\beta+2)^2(\beta+4)(\beta+5)}{(\beta+3)(3\beta+8)(3\beta+10)(7\beta+18)}$
\\ \hline\vphantom{$\Big($} 
$B_{7,9}^*$  &  $\frac{2(\beta+2)^2(\beta+5)}{(\beta+3)(3\beta+8)(3\beta+10)}$
 \\ \hline\vphantom{$\Big($} 
$B_{7,10}^*$ &  $\frac{20(\beta+2)(\beta+3)}{3(3\beta+10)(7\beta+18)}$
\\ \hline\vphantom{$\Big($} 
$B_{7,11}^*$ &   $\frac{5(\beta+2)^2}{(3\beta+10)(7\beta+18)}$

\\ \hline \hline  \vphantom{$\Big($} 
$B_{8,13}^*$ &   $\frac{504(\beta+2)^2(\beta+4)^2(\beta+5)^2(\beta+6)}{ (127\beta^3+1383\beta^2+4958\beta+5880)(31\beta^2+194\beta+300)(3\beta+8)(7\beta+18)}$
\\ \hline\vphantom{$\Big($} 
$B_{8,14}^*$ &   $\frac{336(\beta+2)^2(\beta+4)(\beta+5)^2(\beta+6)}{ (127\beta^3+1383\beta^2+4958\beta+5880)(31\beta^2+194\beta+300)(3\beta+8)}$
\\ \hline\vphantom{$\Big($} 
$B_{8,15}^*$ &  $\frac{1680(\beta+2)^2(\beta+3)(\beta+4)(\beta+5)(\beta+6)}{ (127\beta^3+1383\beta^2+4958\beta+5880)(31\beta^2+194\beta+300)(7\beta+18)}$
 \\ \hline\vphantom{$\Big($} 
$B_{8,16}^*$ &   $\frac{1260(\beta+2)^3(\beta+4)(\beta+5)(\beta+6)}{ (127\beta^3+1383\beta^2+4958\beta+5880)(31\beta^2+194\beta+300)(7\beta+18)}$
\\ \hline\vphantom{$\Big($} 
$B_{8,17}^*$ &  $\frac{280(\beta+2)^2(\beta+3)(\beta+5)(\beta+6)}{ (127\beta^3+1383\beta^2+4958\beta+5880)(31\beta^2+194\beta+300)}$
 \\ \hline\vphantom{$\Big($} 
$B_{8,21}^*$ &  $\frac{560(\beta+2)(\beta+3)^3(\beta+4)}{ (127\beta^3+1383\beta^2+4958\beta+5880)(7\beta+18)^2}$
\\ \hline\vphantom{$\Big($} 
$B_{8,22}^*$ &   $\frac{840(\beta+2)^2(\beta+3)^2(\beta+4)}{ (127\beta^3+1383\beta^2+4958\beta+5880)(7\beta+18)^2}$ 
\\ \hline\vphantom{$\Big($} 
$B_{8,23}^*$ &   $\frac{315(\beta+2)^3(\beta+3)(\beta+4)}{ (127\beta^3+1383\beta^2+4958\beta+5880)(7\beta+18)^2}$
\\ \hline

\end{tabular}
\end{center}
\caption{\label{table:3varbin} Probabilities under the $\beta$-model of the trees involved in the formula for the variance of $\mathit{rQIB}_n$}
\end{table}

\end{document}